\documentclass[11pt]{article}%
\usepackage{amsfonts}
\usepackage{amsmath}
\usepackage{amssymb}
\usepackage{graphicx}
\usepackage{geometry}%
\usepackage{xcolor}
\usepackage{epstopdf}
\usepackage[export]{adjustbox}
\usepackage{subcaption}
\usepackage{bm}
\usepackage{mathrsfs}%
\usepackage{bbm}
\usepackage{natbib}
\usepackage{mathtools}
\usepackage[font={small}]{caption}
\usepackage[titletoc,title]{appendix}
\usepackage{dsfont}
\usepackage{makecell}
\usepackage{enumitem}
\usepackage{multirow}

\providecommand{\U}[1]{\protect\rule{.1in}{.1in}}
\newtheorem{theorem}{Theorem}[section]
\newtheorem{proposition}[theorem]{Proposition}
\newtheorem{lemma}[theorem]{Lemma}

\newtheorem{remark}{Remark}[section]
\newenvironment{proof}[1][Proof]{\noindent\textbf{#1.} }{\ \rule{0.5em}{0.5em}}

\numberwithin{equation}{section}

\normalsize
\DeclareRobustCommand{\lcroof}[1]{
  \hbox{\vtop{\vbox{%
      \hrule\kern 1pt\hbox{%
        $\scriptstyle #1$%
        \kern 1pt}}\kern1pt}%
    \vrule\kern1pt}}

\begin{document}

\author{Weihong Ni $^1$, Corina Constantinescu $^2$, Alfredo D.~Eg\'{\i}dio dos Reis $^3$ \\
	\& V\'eronique Maume-Deschamps $^{4,}$ 
}

\title{Pricing foreseeable and unforeseeable risks in insurance portfolios}
\date{}
\maketitle

\begin{center}
	{\it
		$^1$ Department of Computer Science and Mathematics,
		Arcadia University, {PA}, USA; {niw@arcadia.edu} 
	\\
		$^2$  Institute for Financial and Actuarial Mathematics, Department of Mathematical Sciences, University of Liverpool, UK; c.constantinescu@liverpool.ac.uk
		\\
		$^3$ ISEG and CEMAPRE, Universidade de Lisboa, {PT};
		alfredo@iseg.ulisboa.pt
		\\
		$^4$ {Univ Lyon, Universit\'e Claude Bernard Lyon 1, CNRS UMR 5208, Institut Camille Jordan, F-69622 Villeurbanne, FR; veronique.maume@univ-lyon1.fr}
		\\
	}
\end{center}
\vspace{0.5cm}

\begin{abstract}
%
%
In this manuscript we propose a method for pricing insurance products that cover not only traditional risks, but also unforeseen ones. By considering the Poisson process parameter to be a mixed random variable, we capture the heterogeneity of foreseeable and unforeseeable risks. To illustrate, we estimate the weights for the two risk streams for a real dataset from a Portuguese insurer. To calculate the premium, we set the frequency and severity as distributions that belong to the linear exponential family. Under a Bayesian set-up, we show that when working with a finite mixture of conjugate priors, the premium can be estimated by a mixture of posterior means, with updated parameters, depending on claim histories. We emphasise the riskiness of the unforeseeable trend, by choosing heavy-tailed distributions. After estimating distribution parameters involved using the Expectation-Maximization algorithm, we found that Bayesian premiums derived are more reactive to claim trends than traditional ones.
\\

\end{abstract}

\section{Introduction}
The prime objective of this paper is to {develop the parameter} estimation and premium calculation to the extended model introduced by \cite{li2015risk}, where the so-called unforeseeable 
risks were additionally taken into account. These risks refer to those that are not really predictable due to lack of history, for example, a very real situation during the Spring of 2020, 
a pandemic. {Another example are the emerging risks that have been evolving with the ever smarter way of living. The evaluation of these risks could be different from classical methods. For instance, with the introduction of 
autonomous cars, it would possibly change the magnitude of risks on roads, yet we do not know in which direction precisely, even though they were designed to reduce \textit{human error}.} To model 
such unforeseeable features, a defective random variable was incorporated in the risk model by \cite{li2015risk}, inspired in \cite{dubey77}. There, unforeseeable risks were reflected in the claim 
frequencies only. We extend to severity risks.  

We propose a modification on the classical risk model which accounts {for} the accommodation of the unforeseeable risks.  \cite{li2015risk} show such modified model, a developed version of the classical Cram\'er-Lundberg's risk model introduced by \cite{dubey77}, where the Poisson parameter $\lambda$ is a realization of a random variable $ \Lambda $. The premium income is based on an estimation of $ \Lambda $, for instance the posterior mean $\hat{\lambda}(t)= \mathbb{E}[\Lambda|N(t)] $, namely the risk process is driven by the equation 
\begin{equation}
U(t)=u+c \int_{0}^{t} \hat{\lambda}(s)ds -S(t) \,, \; \; t\geq 0\,.
\label{eq:model2}
\end{equation}
Here $ S(t)=\sum^{N(t)}_{j=0}Y_j $ is the aggregate claims up to time $t$, $ u=U(0) $ is the initial surplus,  $ \{Y_j\}_{j=1}^{\infty} $ a sequence of independent and identically distributed 
(briefly, i.i.d$ . $) random variables with common distribution $H_Y(.)$, existing mean $\mu=\mathbb{E}[Y_1]$. {Moreover, } $Y_0 \equiv 0$, and $ \{N(t), t\geq 0\} $ is a Poisson process 
with  intensity $ \lambda$ $(>0)$. Also, $ \{Y_j\}_{j=1}^{\infty} $ is independent of $ \{N(t), t\geq 0\} $. Here $c$ is the time-invariant component of the premium income, with 
$c=(1+\theta)\mathbb{E}[Y_1]$, 
where $\theta$ $(>0)$ is the loading coefficient. 

 \cite{dubey77} considered a positive probability  $\mathbb{P}[\Lambda=0]=p\,\, (>0)$, so that the counting process $ \{N(t)\} $ is a mixed Poisson process conditional on $\Lambda>0$. \cite{li2015risk} generalised the process to a sum of two mixed counting processes, $ \{N(t)= N^{(1)}(t)+ N^{(2)}(t)\}$, where $ \{N^{(1)}(t)\}$ represents the \textit{historical risk stream} of claim counts and  $ \{N^{(2)}(t)\}$ represents the corresponding \textit{unforeseeable risk stream}. For each, randomised intensity parameters are denoted as $\Lambda^{(1)}$ and $\Lambda^{(2)}$, respectively. This is the process we will consider in the rest of the paper. 
 The randomised intensity parameter of  $N^{(1)}(t)$ has a classical behaviour, meaning that it is a positive random variable, or  
$\mathbb{P}\{\Lambda^{(1)} = 0\} = 0$, unlike  in the second process
where 
$\mathbb{P}\{\Lambda^{(2)} = 0\} = p>0$.
 As in \cite{li2015risk}, we set $\hat{\lambda}(t)= \mathbb{E}[\Lambda^{(1)}+\Lambda^{(2)}|N(t)] $ in (\ref{eq:model2}). 

Not much interpretation is given by the authors, \cite{dubey77} and \cite{li2015risk}, to the situation where $p>0$. {We remark that the second stream can be riskier in either severity or frequency or both. If we think of the development of autonomous cars, experiments suggest that they are potentially safer than human driving if you are looking at the accident rates (Waymo Team, 2016\nocite{waimo216}). But once an accident occurs, it could bring in higher damage than that would have been caused by a human (\cite{McCausland2019}). On the other hand, the outbreak of the asbestos crisis in the 1980s raised no awareness until 20 years later when claims reached more than double the original amount (\cite{white2003understanding}). If a model existed at that time to capture the potential risks embedded in the gradually increasing claims earlier, insurance companies could have avoided dramatic losses.}
In short, the introduction of an unforeseeable stream is to accommodate potential uncertainties, either riskier or less risky, whose nature one cannot foresee exactly {\it if and/or when}, and that the actuary feels it's wise to be reflected in the premium. A new virus pandemic could be an example, signs may have been appeared before.

Modelling two different streams of risks for the same portfolio can be done exclusively on the claim count process as \cite{dubey77} and \cite{li2015risk} did, or in the claim severity, or both in claim counts and severity. In this case we consider there is some sort of dependence between $ \{ N(t), t\geq 0 \} $
and sequence $ \{ Y_j \}^{\infty}_{j=1} $. The latter means that the different streams may \textit{bring} diverse severity behaviour. This is not dealt by any of the above cited authors. 

Letting premium adjust according to claim numbers somehow imitates the operation of a Bonus-malus system (BMS). {For instance, \cite{ni2014bonus} worked with Negative Binomial distributed frequencies and Weibull distributed severities, as an extension to classical settings. Another recent work by \cite{cheung2019note} considered a dependence structure between the claim frequency and severity components via conjugate bivariate prior and calculated Bayesian premiums. In this work,} we also incorporate the two risk streams into the claim severity, add appropriate premium calculation and estimation. 
This is important in order to reflect both risk effects in the premium, for fairness and ruin probability compensation.  For the premium calculation/estimation credibility theory is used, parameter estimation procedure uses the Expectation-Maximization  (EM) algorithm.
 EM algorithm is developed in \cite{dempster1977maximum}, a \textit{guided tour} can be found in \cite{couvreur1997algorithm}. We will be using mixture distributions, particular algorithm application for these can be found in 
\cite{redner1984mixture}. 

The use of Bayesian and credibility methods for posterior premium appear as a natural choice (we have set above the posterior mean $\hat{\lambda}(t)$, for instance).  Markovian bonus-malus methods are easy and common choices for posterior rating based on claim counts only and are used for individualise premia. This is not our problem, besides, adding the severity component it necessarily drives actuaries for credibility estimating methods.  Introductory notions on credibility can be found in \cite{klugman2012loss}, Chapters~17-19, and more advanced in \cite{buhlmann2006course}. 

The manuscript is organised as follows. In Section 2 we work only with the claim frequency component of the model, while in Section~\ref{s:sev_comp} we introduce the claim severity, assuming that 
severity brings some information on the stream origin. In Section~\ref{sec:Bayesian} we deal with the premium calculation/estimation via a Bayesian approach. Section~\ref{s:estima} is 
devoted to the parameter estimation via an EM algorithm, on a concrete dataset. {Claim severity data are randomly generated by the chosen mixed distributions considering both non-extreme and extreme cases. Bayesian premiums are computed in the end using the estimated parameters.} Section~\ref{sec:global} presents discussion on a
\textit{global} likelihood procedure to estimate together the claim frequency and severity distributions, 
and in the last section we write some concluding remarks. {Proofs of 
lemmas are given in the appendix.}
\section{Claim frequency component} \label{s:claim_comp}
The total claim numbers over a time period $(0,t]$ within the portfolio is given by 
\begin{equation}
\label{eq:jointprocess}
N(t) = N^{(1)}(t)+N^{(2)}(t)\,.
\end{equation}
Here $\{N^{(1)}(t)\}$ is a mixed Poisson process  with parameter $\Lambda^{(1)}$ and $\{N^{(2)}(t)\}$ is mixed Poisson conditional on $\{\Lambda^{(2)} > 0\}$, meaning the intensity must be strictly positive. The event $\{\Lambda^{(2)} = 0\}$ implies that $\{N^{(2)}(t) = 0\}$, with probability one, since %
  $$\lim_{\lambda\downarrow 0}\mathbb{P}\{N^{(2)}(t)=0|\Lambda^{(2)}=\lambda\}=\lim_{\lambda\downarrow 0}e^{-\lambda\,t}=1\,.$$ 
$\{N(t)\}$ is also a mixed Poisson process with intensity $\Lambda^{(1)}+\Lambda^{(2)}$ given that $N^{(1)}(t)$ and $N^{(2)}(t)$ are independent, see \cite{li2015risk} [Lemma~1]. Under this global process, $\{N(t)=N^{(1)}(t)+N^{(2)}(t)\}$, we have that
when a claim arrives it comes either from the process $ \{N(t)^{(1)}\}$,   
or 
from the process 
$\{N^{(2)}(t)\}$. We note that in practice actuaries observe the realisations of the overall claim process, which means they \textit{guess} which process the  claim belongs to. 

{Let us} define and denote by $\Xi$, the random variable representing the \textit{split rate} in favour of process $\{N^{(1)(t)}\}$ in $\{N(t)\}$. It is well known that for independent Poisson 
processes, in our case for given positive rates $\Lambda^{(1)}=\lambda^{(1)}$ and $\Lambda^{(2)}=\lambda^{(2)}>0$, we have that the \textit{split} between processes ``(1)'' and ``(2)''  is going to be 
given by  probabilities, 
\begin{eqnarray}
\label{eq:epsilon}	
\xi&=&\dfrac{\lambda^{(1)}}{\lambda^{(1)}+\lambda^{(2)}} \;\;\;\, 
 \text{and } \;\;\;\,	1- \xi 
 = 
 \dfrac{\lambda^{(2)}}{\lambda^{(1)}+\lambda^{(2)}}\,, 
\end{eqnarray}
 respectively.
{Unconditionally, considering that $\xi$ is a particular outcome of the random variable $\Xi$, then we can write, \textit{extending} the unconditional mixed process as defined in \eqref{eq:jointprocess}}
 \begin{eqnarray}	
 \label{eq:Xi}
 \Xi&=&\dfrac{\Lambda^{(1)}}{\Lambda^{(1)}+\Lambda^{(2)}} \,.  
  \end{eqnarray}
 Note that
  \begin{equation}
  \label{eq:probxi1}
 \{\Lambda^{(2)}= 0\} \Leftrightarrow \{\Xi=1\}  \Rightarrow \mathbb{P}[\Lambda=0]=\mathbb{P}[\Xi=1]=p\, .
 \end{equation}
 %
 We denote the distribution function of $\Xi$ as $A_\Xi(.)$. 
 
We will now set further assumptions. From now on, we assume that the Poisson intensities follow gamma distributions: $\Lambda^{(1)}\sim Gamma(\alpha_1, \beta_1)$ and, conditional on $ \Lambda^{(2)}>0 $, $\Lambda^{(2)}|\Lambda^{(2)}>0\sim Gamma(\alpha_2, \beta_2)$. For simplicity, in the sequel we denote $\Lambda^{(2)}_+=\Lambda^{(2)}|\Lambda^{(2)}>0$. {Furthermore, without loss of generality, we consider that the number of claims observed within one unit period $N(1)$ to be denoted as $N$.}
Under a more restrictive assumption for the gamma distributions above, we can arrive at the following lemma:
\begin{lemma} \label{l:mix_nb}
If $\beta_1 = \beta_2 = \beta$, then  the random {variable $N=N(1)$} is represented by a mixture of two Negative Binomial random variables, with probability function
\begin{eqnarray*}
\mathbb{P}(N = n) &= & p\binom{n+\alpha_1-1}{n}\left(\frac{\beta}{\beta + 1}\right)^{\alpha_1}\left(\frac{1}{\beta + 1}\right)^n \\
& & \quad \quad +\,  (1-p)\binom{n+\alpha_1+\alpha_2-1}{n}\left(\frac{\beta}{\beta + 1}\right)^{\alpha_1+\alpha_2}\left(\frac{1}{\beta + 1}\right)^n.
\end{eqnarray*}
\end{lemma}
\begin{remark}
Note that $N^{(2)}(t)$ follows a `Zero Modified' Negative Binomial distribution, briefly, $N^{(2)}(t)\sim ZM\,\, Negative\text{ } Binomial (\alpha_2,\beta)$. It belongs to the  $(a,b,1)$ recursion class of distributions, see Section~6.6 of \cite{klugman2012loss}. Since
\begin{eqnarray*}
M_{N^{(2)}}(\rho) &=
& \mathbb{E}\left[e^{N^{(2)}}|\Lambda^{(2)}=0\right]\times \mathbb{P}\left[ \Lambda^{(2)}=0
 \right]	
+ \mathbb{E}\left[e^{N^{(2)}}|\Lambda^{(2)}>0\right]\times \mathbb{P}\left[ \Lambda^{(2)}>0\right]\\
 &=
& p+ (1-p) \left(\frac{1-\frac{1}{1+\beta}}{1-e^\rho \frac{1}{1+\beta}}\right)^{\alpha_2}\,,
\end{eqnarray*}
which is a weighted average of the pgf's of  a degenerate distribution at $\{0\}$ and a Negative Binomial, member of the $(a,b,0)$ class. \hfill{$\Box$}
\end{remark}

In the sequel, instead of focusing on the claim arrival process $\{N(t)\}$ which consists of combining two counting processes, $\{N^{(1)}(t)\}$ and $\{N^{(2)}(t)\}$, as explained above, we could model the underlying claim counts via mixing random variables, mixing two Gamma random variables with an independent Bernoulli random variable. See \citep{Dominik},
\begin{lemma} \label{l:dominik}
Assume that the randomized Poisson parameter $\Lambda$ follows a prior distribution as a mixture of two Gamma random variables, such that
\begin{equation}
\Lambda = I\, Z_1 + (1-I)\, Z_2,
\end{equation}
where $Z_1\sim Gamma(\alpha_1, \beta)$ and $Z_2\sim Gamma(\alpha_1+\alpha_2, \beta)$ and $I$ is a Bernoulli$(p)$ random variable independent of $Z_1$ and $Z_2$.
Then the distribution  {of $N=N(1)$} is the mixture of two Negative Binomial random variables as shown in Lemma~\ref{l:mix_nb} above. \hfill{}
\label{l:mix_nb2}
\end{lemma}
%

%
Based on the above Bayesian set-up, considering a prior distribution of Gamma mixtures, we will be able to compute easily a corresponding posterior distribution, for a given distribution of an observable random variable. This is important for experience rating where we can use naturally the theory of Bayesian premium, for making posterior estimation of premia, as well as credibility theory, a particular case in Bayesian premium estimation. This is going to be done  in Section~\ref{sec:Bayesian}.
In fact, the equivalence of the previous two constructions can be explained further. Consider the following remark.
\begin{remark}
	Let $ Z_1=\Lambda^{(1)} $, $ Z_2 = \Lambda^{(1)}+\Lambda^{(2)}_+ $, where $\Lambda^{(2)}_+=\Lambda^{(2)}|\Lambda^{(2)}>0$, as denoted before. Also, let $\Lambda^{(2)} = (1-I)\,\Lambda^{(2)}_+ + I \Lambda^{(2)}_0= (1-I)\,\Lambda^{(2)}_+$, with $\Lambda_0^{(2)} = 0$ and $I\sim \text{Bernoulli}(p)$. 
	We have,
\begin{eqnarray*}
\Lambda &=& I\, Z_1 + (1-I)\, Z_2 = I\,\Lambda^{(1)}+(1-I)\,(\Lambda^{(1)}+\Lambda^{(2)}_+) \\
&=& \Lambda^{(1)} + (1-I) \,\Lambda^{(2)}_+ = \Lambda^{(1)}+\Lambda^{(2)}\,.
\end{eqnarray*}
Since $N(t)|\Lambda \sim$ \text{Poisson} $(\Lambda^{(1)}+\Lambda^{(2)})$, it is true that $N(t) = N^{(1)}(t)+N^{(2)}(t)$ where $N^{(1)}(t)|\Lambda^{(1)}\sim  \text{Poisson}(\Lambda^{(1)})$ and $N^{(2)}(t)|\Lambda^{(2)}_+\sim \text{Poisson}(\Lambda^{(2)}_+)$. \hfill{$\Box$}
\end{remark}

In the model presented so far, we distinguish the \textit{unforeseeable} stream of risks from the \textit{historical} stream by considering two different, and independent, claim counting processes, where the randomized intensities are of different nature. 
%
For now, these parameters only influence the claim numbers, not the claim sizes. 
However, we should also consider the possibility that the claims are affected by the risk type. The introduction of a randomness in the intensities can bring a dependence between number of claims and sizes, setting us away from the classical risk consideration that the claim arrival process is independent of the claim severity sequence. We may assume a twofold approach
\begin{enumerate}
\item At a starting stage, we can consider that for a given $\Lambda^{(i)}=\lambda^{(i)}$, $i=1,2$, we have conditional independence. We mean, for a given ($i$) and $\lambda^{(i)}$, severities and claim counts are independent;
\item At a later stage, an {extended} model with dependence.
\end{enumerate} 

Again, we remark that parameters cannot be observed, only claims can. Furthermore, in practice we only observe the realization of the total claim process $\{N(t)=N^{(1)}(t)+N^{(2)}\}$, so that we \textit{guess}/estimate which process each claim arrival belongs to, then make the severity correspondence, if we set a different distribution for each stream.
 Thus, using data and the model as behaving like is defined in Lemma~\ref{l:mix_nb} we have four parameters to estimate, $p$, $\alpha_1$, $\alpha_2$, and $\beta$. Parameters of the severity distributions will be estimated in addition.  

As defined in \eqref{eq:Xi}, $\Xi$ is the random split rate between processes $\{N^{(1)}(t)\}$ and $\{N^{(2)}\}$ in $\{N(t)\}$. Let's consider event $\{\Lambda^{(1)} + \Lambda^{(2)} = \lambda,\,  \Xi = \xi ,\, 0<\xi<1 \}$ as given, so that $\{N(t)\}$ conditionally is a Poisson process. Consider the following remark:
%
 %
%
\begin{remark}
Remark that $\Xi=\xi$ is the probability that an event arriving from process
	$\{N(t)=N^{(1)}+N^{(2)}\}$ is generated by 	$\{N^{(1)}\}$. Conditional on $\{\Lambda^{(1)} + \Lambda^{(2)} = \lambda, \Xi = \xi\}$ we have that $\Lambda^{(1)}$ is also given. Then $ \Lambda^{(1)}=\lambda \xi $, and  $ \Lambda^{(2)}=\lambda - \Lambda^{(1)} = \lambda(1-\xi)\,$.
If $\xi=1 \Rightarrow \Lambda^{(2)}=0 \Leftrightarrow \Lambda= \Lambda^{(1)}=\lambda$ and $N^{(2)}(t)=0 \; \forall\, t\geq 0$ almost surely.
%
%
%
\end{remark}
%

%
\section{Claim severity component \label{s:sev_comp}}
For now, we assume that the distribution of the individual claim size depends on the stream type,  either \textit{historical} or \textit{unforeseeable}. We keep assuming that the actuary may not be able to recognize which stream the claim comes from. At least, he cannot be certain. 
Let $F$ and $G$ denote the distributions for claim severities in the \textit{historical} and \textit{unforeseeable} streams, respectively. Consider that the individual claim size, taken at random,  say $Y$, follows a distribution function denoted as $H(y)$.
{\begin{proposition} \label {l:severity}
	For a given claim $Y$, its distribution function, conditional on $\Xi=\xi$, can be represented by
	\begin{equation}
	\label{eq:dfYxi}
	\mathbb{P}\{Y\leq y|\Xi=
\xi \}:= H_{\xi}(y) = \xi F(y) + (1-\xi)G(y),
	\end{equation}
	where $\Xi = \frac{\Lambda^{(1)}}{\Lambda^{(1)}+\Lambda^{(2)}}\in(0,1]$, and $F, G$ correspond to the distributions for claim severities in the historical and unforeseeable streams, respectively. The set $\{\Xi = 1\} = \{\Lambda^{(2)} = 0\}$ has a probability measure $p$.
\end{proposition}
\begin{proof}
	The proof is straightforward using the law of total probability.
	\begin{eqnarray*}
		\mathbb{P}\{Y\leq y|\xi\} &=& \mathbb{P}\{Y\leq y | Y = Y^{(1)}, \xi\}\mathbb{P}\{Y = Y^{(1)}| \xi\} + \mathbb{P}\{Y\leq y | Y = Y^{(2)}, \xi\}\mathbb{P}\{Y = Y^{(2)}| \xi\}\\
		&=& F(y)\mathbb{P}\{Y = Y^{(1)}|\xi \} + G(y)\mathbb{P}\{Y = Y^{(2)}| \xi\}\\
		&=&F(y)\cdot\xi + G(y)\cdot(1-\xi).
	\end{eqnarray*}
	where $Y^{(1)}$ and  $Y^{(2)}$ denote random variables of the size of claims stemming from Stream~1 (the \textit{historical stream}) and Stream~2 (the \textit{unforeseeable stream}) 
respectively.  Note that when $\Xi = 1$, i.e., $\Lambda^{(2)}=0$, $N(t) = N^{(1)}(t)$ and the probability of having a claim from Stream 1, i.e., $\mathbb{P}\{Y = Y^{(1)}|\xi =1 \} = 1$, where as 
$\mathbb{P}\{Y = Y^{(2)}| \Xi=1\}=0$. This special situation does not affect Equality~\eqref{eq:dfYxi}. \hfill
\end{proof}}

It has been illustrated in \cite{li2015risk}, see their Lemmas~2 and~3,  that the independence between $\frac{\Lambda^{(1)}}{\Lambda^{(1)}+\Lambda^{(2)}_+}$ 
and $\Lambda^{(1)}+\Lambda^{(2)}_+$, 
conditional on $\Xi \neq 1$, 
result in $\Xi | {\Xi \neq 1}\,$ being distributed as  a $Beta (\alpha_1, \alpha_2)$ law. Then, it is easy to derive the unconditional distribution of $Y$ under these assumptions.
{\begin{proposition}
\label{l:beta}
Assume that $\frac{\Lambda^{(1)}}{\Lambda^{(1)}+\Lambda^{(2)}_+}$ and $\Lambda^{(1)}+\Lambda_+^{(2)}$ are independent, and that $\Lambda^{(1)} \sim Gamma (\alpha_1, \beta)$ and $\Lambda_+^{(2)} \sim Gamma(\alpha_2, \beta)$. A given claim size is distributed according to a mixture law of $F(\cdot)$ and $G(\cdot)$, i.e., with density,
\begin{equation}
h_Y(y) = \nu\cdot f(y) + (1-\nu)\cdot g(y),
\label{eqn:clmdist}
\end{equation}
where $\nu = p+(1-p)\frac{B(\alpha_1+1, \alpha_2)}{B(\alpha_1, \alpha_2)}$.
\end{proposition}
\begin{proof}
Under the assumptions, we know that $\Xi| {\Xi\neq 1} \frown Beta (\alpha_1, \alpha_2)$ according to \cite{lukacs1955characterization}'s proportion-sum independence theorem. The subsequent proof is 
then straightforward.  $A_\Xi (\xi)$ denotes the distribution function of $\Xi$  (a mixture), we have
\begin{eqnarray*}
	\mathbb{P}\{Y\leq y\} &=& \int_{(0,1]} H_{\xi}(y) d A(\xi) = pF(y) + (1-p)\int_{(0,1)}H_{\xi}(y)\cdot \frac{\xi^{\alpha_1-1}(1-\xi)^{\alpha_2-1}}{B(\alpha_1, \alpha_2)}d\xi\\
	&=&pF(y) + (1-p)\left[\frac{B(\alpha_1+1, \alpha_2)}{B(\alpha_1, \alpha_2)}\cdot F(y) + \frac{B(\alpha_1, \alpha_2+1)}{B(\alpha_1, \alpha_2)}\cdot G(y)\right]\\
	&=&\left[p+(1-p)\frac{B(\alpha_1+1, \alpha_2)}{B(\alpha_1, \alpha_2)}\right]\cdot F(y) + (1-p) \frac{B(\alpha_1, \alpha_2+1)}{B(\alpha_1, \alpha_2)}\cdot G(y)\\
	&=& \nu\cdot F(y) + (1-\nu)\cdot G(y)\,.
\end{eqnarray*}
 \hfill
\end{proof}
}
\begin{remark}
	\label{rem:indep}
Based on the proof of \cite{li2015risk}' Lemma~2 (see their Appendix) we can conclude that, under the conditions of Lemma~\ref{l:beta}, $Y$ and $N$  are independent. This will allow us to consider the parameter estimation using separately the claim frequency and the severity component.
\end{remark} 
%


{In the previous section we specified a distribution for the random variable $N$, as well as its parametrization. The unconditional distribution for $N$ starts from assuming a Poisson distribution, commonly accepted as assumption for claim count data in the line of motor insurance, for instance. We ended up developing a particular mixture of two Negative Binomial distributions, having started from a different Poisson distribution for each stream of risks.  We'll see in Section~\ref{s:estima} that the assumption is reasonable and may fit actual data.}

For the claim severity we will start with a prior assumption that claim severity distribution belong to the {\it linear exponential family}. This  exponential family is a common fit for insurance severity data. First, we assume a particular form for the distribution of individual claim severity as a mixture of two distributions, denoted as $F(\cdot)$ and $G(\cdot)$. They represent the behaviours of the claim severities from the two streams, separately the historical and unforeseeable streams, respectively $F(\cdot)$ and $G(\cdot)$. Then, if we specify $F(\cdot)$ and $G(\cdot)$, we can move for the parameter estimation.  

Using a simple example, from now onwards consider that both the densities of the two streams come from exponentially distributed random variables,
 $Y\sim Exp(\Theta)$, where $\Theta^{-1}$ is the mean, and as discussed, there is a dependence structure embedded in the random parameter  $\Theta$, such that: 
\begin{enumerate}
	\item 	$\Theta = \mu$ with probability~one if we consider the historical stream; 
	\item 
	$\Theta \sim Gamma(\delta, \sigma)$ if we consider the unforeseeable one.
\end{enumerate}
This implies that claims  in the historical stream conform to the exponential distribution with mean $\mu$, whereas those of the unforeseeable stream conform to a Pareto distribution, $Pareto(\delta, \sigma)$. Their respective unconditional densities are
\begin{equation}
\label{eq:sevdist}
f(y) =  \mu e^{-\mu y} \quad \mbox{and}\quad g(y) = \frac{\delta \sigma^\delta}{(\sigma+y)^{\delta+1}}, \quad \mu, \delta, \sigma > 0\, .
\end{equation}
We chose a mixture of a light tail distribution with a heavy tail one, respectively to represent an \textit{historic} behaviour example and an unforeseeable one.

{
	\begin{remark}
		This consideration can actually be connected with the example presented in Corollary 2 in \cite{li2015risk}. If we let $\zeta_1 = \mu$, treat $\zeta_2$ as a Gamma$(\delta, \sigma)$ random variable and integrate the ruin function obtained there over $\zeta_2$, we could in theory represent the ruin probability for our model in a closed form. 
	\end{remark}
	}

Our model will be complete, and ready for parameter estimation/distribution fit, once the premium calculation is defined. 

\section{Bayesian premium} \label{sec:Bayesian}
\subsection{Premium estimation}\label{ss_premium_est}
Premia are set to be calculated for rating periods, most commonly annual periods. For estimating, risk observations must be organized by periods, suppose that we observed the risk for $m$ years, or $m$ rating periods. Then, our observation vector for the claim counts is ${\bf n} = \{n_1,\ldots, n_m\}$. For each year, say year $j$, we'll have $n_j$ individual claims observed, say vector ${\bf y}_j = \{y_{j1},\ldots, y_{jn_j}\}$. This means that observations in year $j$ are set in the two fold vector $(n_j,{\bf y}_j)$. To calculate the annual premium, one only needs the annual amounts ,on aggregate, specifically the observations of $S(1)_i$, $i=1,\dots,m$, since a premium is set for an aggregate quantity. However that would be insufficient in our case, since we consider that claim counts and the severities bring separate and different information on an unforeseeable stream. Another question would rise: Could we separate information contribution from counts and from severities?

In a classical model the two  annual sequences of claim counts and severities are assumed to be independent and the annual pure premium is given by $E\left[S(1)\right]= E\left[N\right]\times E\left[Y\right] $ and the two premium factors are often estimated separately. In view of Remark~\ref{rem:indep} this is also in our case. So, for now we'll deal separately estimation of $E\left[N\right]$ and $E\left[Y\right] $.

Let's pick up the case of estimation of $E\left[N \right]=E\left[N_{m+1}\right]$ for the upcoming year. The problem and procedures for the severity component  $E\left[Y\right] $ are similar. If we proceed in a \textit{classical way}, we would simply estimate the parameters involving that mean $p,\, \alpha_1,\, \alpha_2,\, \beta, $  and we would get an estimate, say  $\hat{E}\left[N \right]$. However, since we have set the model under a Bayesian set-up (recall Lemma~\ref{l:dominik}), we naturally open the Bayesian methodology of insurance premium calculation, where modern credibility theory is included. See \cite[Sections~18.3-4]{klugman2012loss} The hypothesis of conditional independence and identically distribution past annual observations, ${\bf N} = \{N_1,\ldots, N_m\}$,  is assumed, given $\Lambda= \Lambda^{(1)} + \Lambda^{(2)} =\lambda$. 

The Bayesian (pure) premium is defined as the mean of the predictive distribution $E\left[N_{m+1}|{\bf N}={\bf n}\right]$. An interpretation is simple, parameters cannot be observed, only claims are, 
either counts or severities, and they {\it should} bring some information about unknown parameters. So, $E\left[N_{m+1}|{\bf N}={\bf n}\right]$ is an estimate for the premium we want to calculate. {On 
the other hand, it could be easily shown that $E\left[N_{m+1}|{\bf N}={\bf n}\right]= E\left[\mu_{m+1}(\Lambda)|{\bf N}={\bf n}\right]$, where $\mu_{m+1}(\Lambda)= \mu(\Lambda)=\Lambda$ is the hypothetical mean of the Poisson likelihood.
The former expectation is the mean of the predictive distribution. The latter expectation is the expected value of the hypothetical mean $\mu(\Lambda)$, integrated over the posterior distribution. $\mu(\Lambda)$ is also known in insurance risk theory, and credibility theory, as the (unknown) risk premium. 
{The credibility premium is the best linear approximation for the Bayesian premium in the sense of minimizing the expected squared error (see \cite[Section~20.3.7 and Formula~(20.28)]{klugman2012loss}).}
When we work with linear exponential family likelihood and their conjugate priors (no truncations), we can usually get what is called \textit{Exact Credibility} in the literature, see \cite[Section~18.7]{buhlmann1967experience}, i.e., the Bayesian premium equals the credibility premium. Despite the fact that our setting will not result in the exact credibility, we will see that they can be related.} 

%
 
 Consider a generic a random variable, say $X$, following a distribution belonging to the linear exponential family, depending of some parameter, say $\theta$. Its probability or density function is of the form (see \cite[Chapters~5 and 15]{klugman2012loss})
\begin{equation}
\label{eq:lexpof}
f_{X|\Theta}(x|\theta) = \frac{p(x) e^{r(\theta) x}}{q(\theta)}\, .
\end{equation}
If parameter $\theta$ is considered to be an outcome of a random variable $\Theta$, whose prior density function is of the following form, denoted as $\pi(.)$
\begin{equation}
\label{eq:lexpoprior}
\pi(\theta) = \frac{[q(\theta)]^{-k} e^{\mu k r(\theta) } r'(\theta)}{c(\mu, k)}\,,
\end{equation} 
then we are in the presence of a conjugate prior. In the above formulae, \eqref{eq:lexpof} and \eqref{eq:lexpoprior}, $p(.)$ is some function not depending {on the} parameter $\theta$. {On 
one other hand,} $r(.)$, $q(.)$ are functions of the parameter; $c(.)$ is a normalizing function of given parameters. 

It is clear from Sections~\ref{s:claim_comp} and \ref{s:sev_comp} that we work  with situations where 
the risk random variables that are observable follow  distributions that are members of the linear exponential family given by~\eqref{eq:lexpof}.  Associated to these we consider priors that {are  
mixture distributions} of the form~\eqref{eq:lexpoprior}. In this situation we don't have { exact} credibility but
 some mixture of { exact} credibility situations, so that {the} premium estimation is simpler than we would expect, with easy interpretation.

{For} claim counts we consider a Poisson random variable, with random parameter and distributed as {a mixture} \eqref{eqn:gammamix} (in the appendix). For the claim severities, we consider 
an exponential random variable whose parameter is random, distributed as  a mixture of a degenerate distribution and a Gamma (see end of Section~\ref{s:sev_comp}). 

As a general case, consider that the prior distribution is a finite mixture of distributions of family \eqref{eq:lexpoprior} such that 
\begin{equation}
\pi(\theta) = \sum_{i = 1}^\eta \omega_i \frac{[q(\theta)]^{-k_i} e^{\mu_i k_i r(\theta)} r'(\theta)}{c_i(\mu_i, k_i)}\,,
\label{eqn:mixprior}
\end{equation}
where $\omega_i,\, i=1,\cdots, \eta$, are given weights.
\begin{theorem}
		\label{cor:conjmix}
	Suppose that given $\Theta = \theta$, the observable vector of i.i.d.~random variables ${\bf X}= (X_1,\ldots, X_n) $ have common probability function given by \eqref{eq:lexpof}, and that the  
	prior distribution of $\Theta$, $ \pi(\theta) $, is of the form given by \eqref{eqn:mixprior}. Then, the posterior distribution, denoted as $\pi(\theta | {\bf x})$ with ${\bf x} = \{x_1, \ldots, x_n\}$, is of a mixture form as \eqref{eqn:mixprior}:
		\begin{eqnarray}
	\pi(\theta | {\bf x}) &=& \sum_i \omega^*_i \frac{[q(\theta)]^{-k^*_i} e^{\mu^*_i k^*_i r(\theta)} r'(\theta)}{c_i(\mu^*_i, k^*_i)},
	\end{eqnarray}
	where
	\begin{eqnarray}
	\mu^*_i &=& \frac{\mu_i k_i+\sum_j x_j}{k_i+n},\\
	k^*_i &=& k_i+n,\\[0.25cm]
	%
	%
	w^*_i &=& \frac{\omega_i \frac{c_i(\mu^*_i, k^*_i)}{c_i(\mu_i, k_i)}}{\sum_i \omega_i \frac{c_i(\mu^*_i, k^*_i)}{c_i(\mu_i, k_i)}}\,.
	\end{eqnarray}
\end{theorem}
\begin{proof}
	With observations ${\bf X} = {\bf x}$, the posterior distribution is
	\begin{eqnarray*}
	\pi(\theta | {\bf x}) &=& \frac{ \frac{\prod_j p(x_j)e^{r(\theta) \sum_j x_j} }{[q(\theta)]^n} \cdot \sum_{i = 1}^m \omega_i \frac{[q(\theta)]^{-k_i} e^{\mu_i k_i r(\theta)} r'(\theta)}{c_i(\mu_i, k_i)} }{\displaystyle\int_{\Theta} \left(\frac{\prod_j p(x_j)e^{r(\theta) \sum_j x_j} }{[q(\theta)]^n} \cdot \sum_{i = 1}^m \omega_i \frac{[q(\theta)]^{-k_i} e^{\mu_i k_i r(\theta)} r'(\theta)}{c_i(\mu_i, k_i)} \right) \,\rm d\theta} \\[0.25cm]
	%
	%
	&=&\frac{\sum_i \omega_i \frac{[q(\theta)]^{-k_i-n} e^{(\mu_i k_i+\sum_j x_j) r(\theta)} r'(\theta)}{c_i(\mu_i, k_i)}}{\sum_i \omega_i \displaystyle{\int_{\Theta}} \frac{[q(\theta)]^{-k_i-n} e^{(\mu_i k_i+\sum_j x_j) r(\theta)} r'(\theta)}{c_i(\mu_i, k_i)} \,\rm d\theta }
	\\[0.25cm]
	%
	%
	&=&\frac{\sum_i \omega_i \frac{[q(\theta)]^{-k_i-n} e^{(\mu_i k_i+\sum_j x_j) r(\theta)} r'(\theta)}{c_i(\mu_i, k_i)} }{\sum_i \omega_i \frac{c_i(\mu^*_i, k^*_i)}{c_i(\mu_i, k_i)}\displaystyle{\int_\Theta} \frac{[q(\theta)]^{-k^*_i} e^{k^*_i\mu^*_i r(\theta)} r'(\theta)}{c_i^*(\mu^*_i, k^*_i)}\,\rm d\theta }  
	%
	%
	= \sum_i \omega^*_i \frac{[q(\theta)]^{-k^*_i} e^{\mu^*_i k^*_i r(\theta)} r'(\theta)}{c_i(\mu^*_i, k^*_i)}.
	\end{eqnarray*}
\hfill
\end{proof}

From the posterior distribution,  we calculate  the mean, and thus we can calculate the Bayesian premium, as an estimate for the risk premium [see \cite{klugman2012loss}, Chapters 17 and 18] defined as $\mathbb{E}[X|\theta]=\mu(\theta)$. The Bayesian premium is given by estimate for next rating period, $n+1$, $\mathbb{E}[X_{n+1}|{\bf X} = {\bf x}]$ (with $n$, ${\bf X}$ and ${\bf x}$  as defined in Theorem~\ref{cor:conjmix}) and can be calculated via the posterior distribution $\pi(\theta | {\bf x})$, as
\begin{equation*}
\mathbb{E}[X_{n+1}|{\bf x}]=\mathbb{E}[\mu(\Theta)|{\bf x}]\,.
\end{equation*}
Whenever this expectation is a linear function on the observations we have exact credibility and we can write { explicitly} the estimate as a weighted average of the form 
\begin{equation}\label{eq_credi_form}
\mathbb{E}[\mu(\Theta)|{\bf x}]= w\,\bar{x}+(1-w)\,\mathbb{E}[\mu(\theta)] \,,
\end{equation}
where $w$ is the weight, function of $n$, $\bar{x}$ is the empirical mean of the observations and $\mathbb{E}[\mu(\theta)]$ is the (known) collective or manual premium.
 
As said in Remark~\ref{rem:indep}, we can calculate the premium components, separately for the claim counts and severity. Thus, next we are going to consider the claim frequency and then the severity component.

%
%
\subsection{Frequency component}
\label{sec:bayfreq}
If we consider a prior to be a Gamma mixture of the form given by \eqref{eqn:gammamix}, we reach to a posterior with a similar mixture, with updated parameters. This is {stated} in 
Lemma~\ref{l:pstGmix} that follows.
\begin{lemma}\label{l:pstGmix}
	Let the prior distribution be \eqref{eqn:gammamix}. With a Poisson model distribution, the posterior distribution is still a mixture of Gamma distributions, with updated parameters:
	\begin{eqnarray}
	\pi(\lambda | {\bf n}) &=& w \, \cdot \frac{ (\beta+m)^{\sum_j n_j +\alpha_1}\lambda^{\sum_j n_j +\alpha_1-1}e^{-(m+\beta) \lambda} }{\Gamma(\sum_j n_j+\alpha_1)} \nonumber
	\\[2.5mm]
	&&+ (1-w)\cdot \frac{ (\beta+m)^{\sum_j n_j +\alpha_1+\alpha_2}\lambda^{\sum_j n_j +\alpha_1+\alpha_2-1}e^{-(m+\beta) \lambda} }{\Gamma(\sum_j n_j+\alpha_1+\alpha_2)},
	\label{eqn:posterior}
	\end{eqnarray}
	where
	\begin{eqnarray}
w &=& \frac{1}{1+G(p, \alpha_1, \alpha_2, \beta, {\bf n})} \label{eq:w}\\[1.75mm]
	G(p, \alpha_1, \alpha_2, \beta, {\bf n}) & = & \frac{1-p}{p}\cdot\frac{B(\alpha_1, \alpha_2)}{B(\sum_j n_j +\alpha_1, \alpha_2)}\cdot\left(\frac{\beta}{\beta+m}\right)^{\alpha_2}, \nonumber 
\end{eqnarray}
	 ${\bf n} = \{n_1,\ldots, n_m\}$ is the vector of claim count observations and $m$ is the sample size ($\sum_j \cdot=  \sum_{j=1}^m \cdot $).
\end{lemma}
 For a fixed individual claim size, The pure premium estimator will be considered as the posterior mean of the claim frequency component, i.e.,
\begin{equation}
\label{eq:BayesPremium_N}
\mathbb{E}[N_{m+1}| {\bf n}] = \mathbb{E}[\Lambda | {\bf n}]= w\cdot \frac{\sum_j n_j +\alpha_1}{\beta+m} + (1-w)\cdot \frac{\sum_j n_j +\alpha_1+\alpha_2}{m+\beta}\,,
\end{equation}
where $w$ is given by \eqref{eq:w}. 

 We note that the premium estimate in \eqref{eq:BayesPremium_N} {does not have the credibility form} as it may look at first sight (see \cite[Chapters 17 and 18]{klugman2012loss}). 
The $G(\cdot)$ function depends on the vector ${\bf n}$, through a Beta function, so we do not have a posterior linearity here. Furthermore, due to 
Stirling's formula we have that the ratio of the Gamma functions 
 $\Gamma(n+a)/\Gamma(n+b)\sim n^{a-b}$ as $n\rightarrow\infty$, 
 allowing us to write that the function 
 $$
 G(p, \alpha_1, \alpha_2, \beta, {\bf n})\sim \frac{1-p}{p}\frac{\Gamma(\alpha_1)}{\Gamma(\alpha_1+\alpha_2)}
 \left( \frac{\beta \sum_j n_j}{\beta+m}\right)^{\alpha_2} \text{ as } \sum_j n_j \rightarrow \infty\,. 
 $$ 
 
 In fact, the Bayesian premium calculated is a mixture of two credibility components. For instance, the first component can be split into form given by \eqref{eq_credi_form}: 
 \begin{equation}
 \label{eq:BayesPremium_N1}
  \frac{\sum_j n_j +\alpha_1}{\beta+m} =
  \frac{m}{\beta+m}\cdot \left(\frac{\sum_j n_j}{m}\right) +
   \frac{\beta}{\beta+m}\cdot  \left(\frac{\alpha_1}{\beta}\right) \,.
 \end{equation}
 {It} means it can be viewed as the credibility premium, as if there were only historical claims. Similarly, $\frac{\sum_j n_j +\alpha_1+\alpha_2}{{m+\beta}}$ gives a corresponding formula when 
considering the existence of an unforeseeable stream. However this component keeps information on the historical/foreseeable stream through quantities $\alpha_1$ and $\beta$. 
\subsection{Severity component}
For the severity component, when claims are exponentially distributed with parameter $\Theta$, random variable,  we  can set the prior as
\begin{equation}
\pi(\theta) = \nu\Delta_\theta(\{\mu\}) + (1-\nu) \frac{\theta^{\delta-1} \sigma^{\delta} e^{-\sigma \theta}}{\Gamma(\delta)}\,, 
\label{eqn:sizemix}
\end{equation}
where $\Delta_\theta(\{\mu\}) = \mathbf{1}_{\{\mu\}}(\theta)$ is the Dirac measure, and $ \nu $ is as given in \eqref{eqn:clmdist}. Given $ \Theta=\theta $, the conditional  distribution of a single severity is exponential with mean $1/\theta$. Thus, the posterior is also a mixture. 

We are calculating a premium estimate for the severity component only, based on the observed single quantities. Let's consider that the sample is generically of  size $m^\ast$ and the observation vector is ${\bf y} = \{y_1,\ldots, y_{m{^\ast}}\}$. If the sample size of claim counts is $m$, then $m^\ast=\sum_{j=1}^{m}n_j\,$, $n_j\geq 1$. The posterior distribution is set in the following lemma and the premium form is given in \eqref{eq:BayesPremium_Y}. 
\begin{lemma}\label{l:sevpost}
	Let the prior distribution be \eqref{eqn:sizemix} and the conditional  distribution of a single severity, given $\Theta=\theta$, be exponential with mean $1/\theta$. The posterior distribution is still a mixture distribution in terms of the prior form with updated parameters, such that
	\begin{equation}
	\pi(\theta | {\bf y}) = \omega\cdot \Delta_\theta(\{\mu\}) + (1-\omega) \cdot \frac{ (\sigma+\sum_i y_i)^{{m^\ast}+\delta}\, \theta^{{m^\ast} +\delta-1} \, e^{-(\sigma+\sum_i y_i) \theta} }{\Gamma({m^\ast}+\delta)},
	\label{eqn:sizeposterior}
	\end{equation}
	where
	\begin{eqnarray}
\omega &=&	\frac{1}{1+\varphi(\nu, \mu, \delta, \sigma, {\bf y})} \label{eq:omega} \\[0.15cm]
	\varphi(\nu, \mu, \delta, \sigma, {\bf y}) & = & \frac{1-\nu}{\nu}\cdot\frac{\Gamma(m^\ast+\delta) \sigma^{{m^\ast}+\delta}}{\Gamma(\delta) (\sigma+\sum_i y_i)^{{m^\ast}+\delta}}\cdot \mu^{-{m^\ast}} e^{\mu \sum_i y_i}\,,\nonumber
	\end{eqnarray}
	with $ \nu $ as given in \eqref{eqn:clmdist},  ${\bf y} = \{y_1,\ldots, y_{m{^\ast}}\}$ and $\sum_i \cdot =\sum_{i=1}^{m^\ast} \cdot$.
\end{lemma}

As in \eqref{eq:BayesPremium_N}, the posterior mean does not exhibit linearity in the observations. The Bayesian estimator of the severity component is given by the weighted average between the mean of the exponential random variable with parameter ($\mu$) and the mean of the  posterior $Gamma({m^\ast}+\delta, \sigma+\sum_i y_i)$ random variable, namely
\begin{equation}
\label{eq:BayesPremium_Y} 
\mathbb{E}[Y_{m+1}| {\bf y}] = \mathbb{E}[\mu(\Theta)|{\bf x}] = \omega \, \frac{1}{\mu} + (1-\omega)\,  \frac{\sum_{i=1}^{m^\ast} y_i+\sigma}{{m^\ast} + \delta - 1}, 
\end{equation}
where $\omega$ is defined in \eqref{eq:omega}. Interpretation of Formula~\eqref{eq:BayesPremium_Y} is similar to that of \eqref{eq:BayesPremium_N}.

\begin{remark}

{ If $m^*= 0$ then by convention, $\displaystyle \sum_i y_i = 0$ and 
\[
\varphi(\nu, \mu, \delta, \sigma, {\bf y}) \sim \frac{1-\nu}{\nu}.
\]
}

\end{remark}

Bayesian premium $P_{m+1}$ is an estimate for the next rating period $m+1$ and it can be derived from \eqref{eq:BayesPremium_N}, \eqref{eq:BayesPremium_Y},  \eqref{eq:w} and \eqref{eq:omega}, as
\begin{eqnarray}
	&&\mathbb{E}[P_{m+1}| {\bf n} ,{\bf y}] 
	= \mathbb{E}[N_{m+1}| {\bf n}] \mathbb{E}[Y_{m+1}| {\bf y}]  \\ \nonumber
	&=& \left( w \cdot \frac{\sum_{j=1}^{m} n_j +\alpha_1}{\beta+m} + (1-w)\cdot \frac{\sum_{j=1}^{m} n_j +\alpha_1+\alpha_2}{m+\beta}\right) \left(\omega \cdot \frac{1}{\mu} + (1-\omega) \cdot \frac{\sum_{j=1}^{m^\ast} y_i+\sigma}{m^\ast + \delta - 1}\right).	
\label{eqn:prem}
\end{eqnarray}
In this formula we have a set of prior parameters that are most often unknown, and they need to be estimated from observed data, leading to the \textit{Empirical Bayesian Premium}. This is not an easy task, especially when we have to deal with mixtures of distributions, meaning a large number of parameters to estimate. 

\section{Parameter estimation, with a real data example} 
\label{s:estima}

 In this section, we explain how we employed the Expectation-Maximization (EM) algorithm to estimate the prior parameters of our model based on a data set from a Portuguese insurer. Unfortunately we could not find a {dataset} having both the claim counts and the corresponding severities. Often data is recorded on {an} aggregate {level}. The implementation of the algorithm {depends} on the distribution choice. We assumed the use of {a mixed parametric distribution for the severity as an illustration for our proposed model}. 
 
 {Since we are dealing with mixed distributions}, the direct use of the Maximum Likelihood Estimation method is {not always computationally tractable}. We employ the EM algorithm, to find  \textit{best fit} estimates. {Generally speaking, it is an iterative} method to find {estimates} for model parameters when dealing with incomplete data, missing data points, or unobserved (hidden) latent variables. See  \cite{dempster1977maximum} and \cite{couvreur1997algorithm}.
{Estimating parameters for mixed distributions can be considered as a case with unobserved latent variables. It works by initializing parameters according to which probabilities for each possib
value of the latent variable can be computed. Then {the probabilities of the latent variable are used} to derive better estimates of the parameters. The process continues until convergence.} 
\subsection{Data Description}
We were provided with a quarterly claim count dataset from a Portughese motor insurance portfolio. Namely, it records the total number of claims arising from the portfolio every quarter and there were a total of approximately 180 quarters recorded. Therefore, we could treat each entry as an observation for $N(t)$ for a fixed $t$ which stands for a quarter here.\textit{We start to write $N$ for short in the sequel since $N(t)$ is a random variable for fixed $t$.} We took care of the Third Party Liability claims only as those are the obligatory components for a car insurance policy. It was assumed that the portfolio is closed over the underlying period and that claim counts for each quarter are independent observations for $N$. 

Table~\ref{t:ctstats} shows a brief summary of the data under consideration. The last two columns show the standard deviation and the coefficient of variation for this dataset. Its distribution can be also visualised in the histogram shown by Figure~\ref{fig:hist1}. Clearly, data indicates a separation around 6000, which serves as a clue for the adoption of a mixture model as we theoretically derived above. Next, we explain the estimation procedure.
\begin{table}[b]
	\center
	\begin{tabular}{cccccccc}
		\hline\hline
		{\bf Min.} & {\bf 1st Qu.} & {\bf Median} & {\bf Mean} & {\bf 3rd Qu.} & {\bf Max.} & {\bf S.D.} & {\bf C.V.}\\
		3685 & 4879 & 5373 & 5552 & 6432 & 7316 & 902.844 & 0.1626\\
		\hline
	\end{tabular}
	\caption{Summary of statistics}
	\label{t:ctstats}
\end{table}
\begin{figure}
	\center
	\includegraphics[scale=0.3]{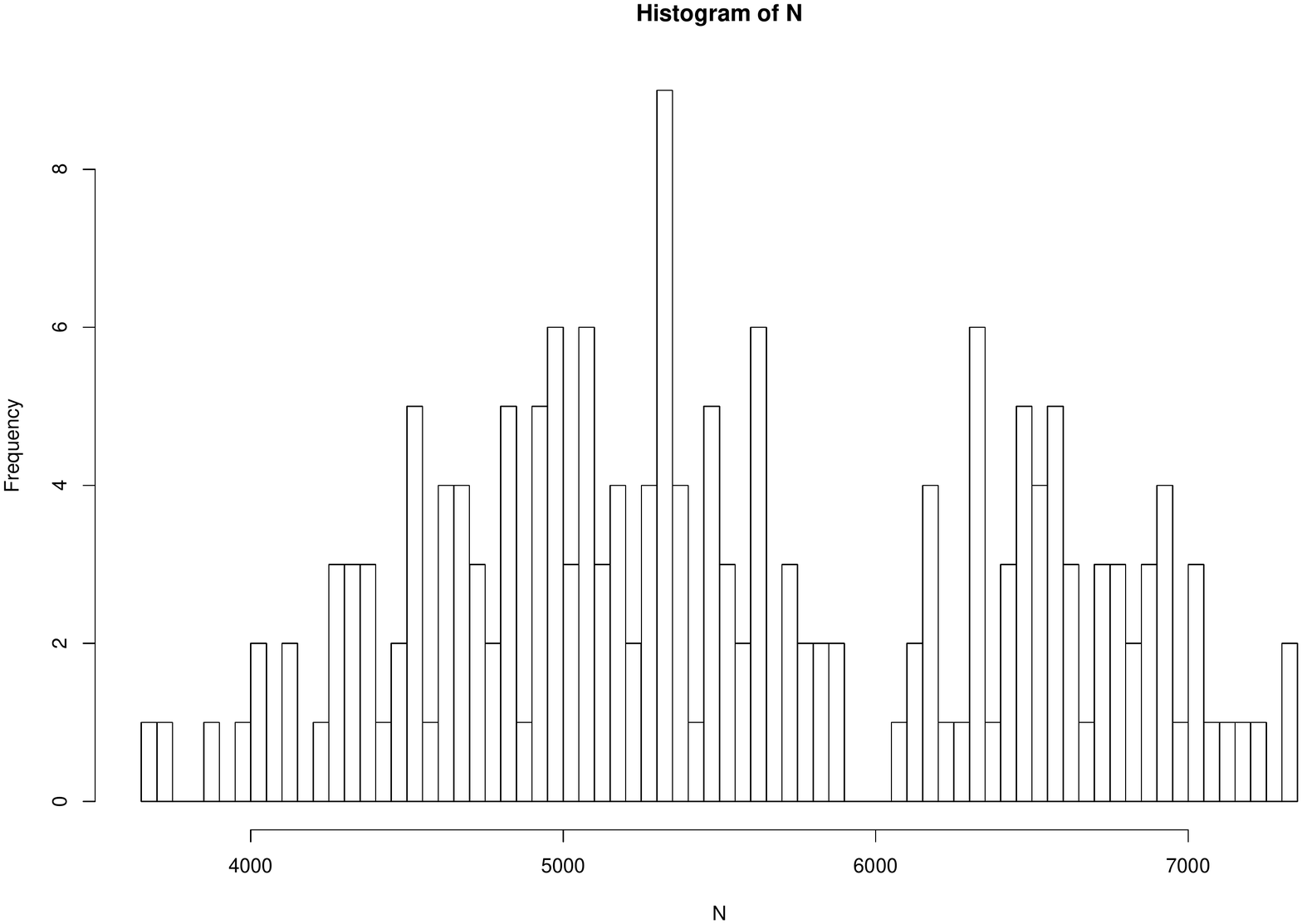}
	\caption{Histogram of claim counts}
	\label{fig:hist1}
\end{figure}
\subsection{Parameter estimation via the EM Algorithm}
\subsubsection*{Claim frequency component}
\label{sec:emfreq}
Recall that we work with a mixture model and that we only observe a global $N_i \in \mathbb{N}$, we mean, we do not know which negative binomial distribution it comes from. Here, we introduce a latent variable, denoted as $Z$, representing the missing information regarding the provenience of each observation. Since we only have  two mixed distributions, $Z\frown Bernoulli(p)$ is a Bernoulli random variable with $\mathbb{P}\{Z = 1\} = q=1-\mathbb{P}\{Z = 0\}$. Now, the complete information is given by the vector $\{N, Z\}$. We can write the complete likelihood function as follows.
\begin{eqnarray*}
L(\theta|N, Z)	&=& \prod_{i=1}^{m} \left[p\binom{n_i+\alpha_1-1}{n_i}\left(\frac{\beta}{\beta +1}\right)^{\alpha_1}\left(\frac{1}{\beta + 1}\right)^{n_i} \right]^{z_i}\\
	& & \hspace{0.5cm}
	 \times\left[(1-p)\binom{n_i+\alpha_1+\alpha_2-1}{n_i}\left(\frac{\beta}{\beta + 1}\right)^{\alpha_1+\alpha_2}\left(\frac{1}{\beta + 1}\right)^{n_i} \right]^{1-z_i}\, ,
\end{eqnarray*}
and according to the law of total probability, we have, 
$$\mathbb{P}\{N\} = \mathbb{P}\{N | Z = 1\}\mathbb{P}\{Z = 1\} + \mathbb{P}\{N | Z = 0\}\mathbb{P}\{Z = 0\}.$$
Here, $\theta \in \{\alpha_1, \alpha_2, \beta, p\}$ are the parameters we are interested in estimating for the claim count distribution. Correspondingly, the log-likelihood function is given by, from above,
\begin{eqnarray}
\mathcal{L}(\theta|N, Z) &=& \sum_i z_i \log\left[p\binom{n_i+\alpha_1-1}{n_i}\left(\frac{\beta}{\beta + 1}\right)^{\alpha_1}\left(\frac{1}{\beta + 1}\right)^{n_i} \right]\nonumber\\
&& \hspace{0.35cm}
+\sum_i (1-z_i)\log \left[(1-p)\binom{n_i+\alpha_1+\alpha_2-1}{n_i}\left(\frac{\beta}{\beta + 1}\right)^{\alpha_1+\alpha_2}\left(\frac{1}{\beta + 1}\right)^{n_i} \right].
\label{eqn:loglikelihood}
\end{eqnarray}
%

We summarize how the algorithm works. It is an iterative process where each iteration consists of two steps, the E-step and the M-step, standing for Expectation and Maximisation, respectively.
\begin{enumerate}
	\item We begin with an initially determined parameter values $\theta^{(0)} = \{\alpha_1^{(0)}, \alpha_2^{(0)}, \beta^{(0)}, p^{(0)}\}$.
	\item E-Step 
	
	For the $(l+1)$-{th} iteration, $l = 0, 1, \dots$ , we first seek for the expected value of $Z_i$ conditional on the observations together with the current parameter estimates $\theta^{(l)}$, i.e., estimates from the previous $l$-{th} iteration. This is denoted as $ \tau_i $,
	\begin{eqnarray*}
		\tau_i &=& \mathbb{E} [Z_i| N,\theta^{(l)}] = 1 \times \mathbb{P}\{Z_i = 1|N_i,\theta^{(l)}\} + 0 \times \mathbb{P}\{Z_i = 0|N_i,\theta^{(l)}\} = \frac{\mathbb{P}\{Z_i = 1, N_i=n_i|\theta^{(l)}\}}{\mathbb{P}\{N_i| \theta^{(l)}\}}\\[0.15cm]
		&=&\frac{p^{(l)}\binom{n_i+\alpha^{(l)}_1-1}{n_i}\left(\frac{\beta^{(l)}}{\beta^{(l)} + 1}\right)^{\alpha^{(l)}_1}\left(\frac{1}{\beta^{(l)} + 1}\right)^{n_i}}{p^{(l)}\binom{n_i+\alpha^{(l)}_1-1}{n_i}\left(\frac{\beta^{(l)}}{\beta^{(l)} + 1}\right)^{\alpha^{(l)}_1}\left(\frac{1}{\beta^{(l)} + 1}\right)^{n_i} + (1-p^{(l)})\binom{n_i+\alpha^{(l)}_1+\alpha^{(l)}_2-1}{n_i}\left(\frac{\beta^{(l)}}{\beta^{(l)} + 1}\right)^{\alpha^{(l)}_1+\alpha^{(l)}_2}\left(\frac{1}{\beta^{(l)} + 1}\right)^{n_i}}
	\end{eqnarray*}
	
	Subsequently, based on  Expression~\eqref{eqn:loglikelihood}, we compute the expectation of the log-likelihood function with respect to the conditional distribution of $Z$, given $N$ under the current estimates $\theta^{(l)}$, denoted as $ Q(\theta |\theta^{(l)}) $. Considering we have we have $m$ independent observations, we have
	\small{\begin{eqnarray}
		&& Q(\theta |\theta^{(l)}) = \mathbb{E}_{Z_i | N_i, \theta^{(l)}}[\mathcal{L}(\theta | N_i,Z_i)]\label{eqn:expectation}\\[0.15cm]
		&=&\sum_{i=1}^{m}\tau_i \left[\log p^{(l)} + \log\binom{n_i+\alpha_1^{(l)}-1}{n_i}\ + \alpha_1^{(l)}\log \frac{\beta^{(l)}}{\beta^{(l)} + 1} + n_i \log \frac{1}{\beta^{(l)} + 1}\right] \nonumber\\
		&& + \sum_{i=1}^{m} (1-\tau_i) \left[\log (1- p^{(l)}) + \log\binom{n_i+\alpha_1^{(l)} + \alpha_2^{(l)} -1}{n_i} + (\alpha_1^{(l)}+\alpha_2^{(l)})\log \frac{\beta^{(l)}}{\beta^{(l)} + 1} + n_i \log \frac{1}{\beta^{(l)} + 1}\right]\nonumber\\[0.15cm]
		&=& \log p^{(l)} \sum_{i}\tau_i + \sum_i\tau_i \log\binom{n_i+\alpha_1^{(l)} -1}{n_i} + \alpha_1^{(l)} m\log \frac{\beta^{(l)}}{\beta^{(l)} + 1}+ \log \frac{1}{\beta^{(l)} + 1}\sum_i n_i\nonumber\\
		&& + \log (1- p^{(l)}) \sum_{i} (1-\tau_i) + \sum_{i} (1-\tau_i)\log\binom{n_i+\alpha_1^{(l)} + \alpha_2^{(l)} -1}{n_i} +\alpha_2^{(l)}\log \frac{\beta^{(l)}}{\beta^{(l)} + 1}\sum_{i} (1-\tau_i).\nonumber
		\end{eqnarray}}
	\item M-Step\\
	The maximisation step is set to find the parameter values that maximises function  \eqref{eqn:expectation} and they become the estimates to be used in the next iteration, i.e.,
	\begin{equation}
	\theta^{(l+1)} =\operatornamewithlimits{argmax}_{\theta} Q(\theta |\theta^{(l)})
	\label{eqn:max}
	\end{equation}
	For this, we take the gradient of $Q(\theta |\theta^{(l)})$, equate to zero and solve for $\{\alpha_1, \alpha_2, \beta, p\}$ simultaneously.
	\begin{eqnarray}
	\frac{\partial Q}{\partial \alpha_1} &=& \sum_i \tau_i [\digamma(n_i+\alpha_1) - \digamma(\alpha_1)] + \sum_i(1-\tau_i)[\digamma(n_i+\alpha_1+\alpha_2) - \digamma(\alpha_1 + \alpha_2)] + m\log\frac{\beta}{1+\beta}=0;\nonumber\\
	\frac{\partial Q}{\partial \alpha_2} &=& \sum_i(1-\tau_i)\left[\digamma(n_i+\alpha_1+\alpha_2) - \digamma(\alpha_1 + \alpha_2)] + \log\frac{\beta}{1+\beta}\right]=0;\nonumber\\
	\frac{\partial Q}{\partial \beta} &=& \frac{\alpha_1 m + \alpha_2\sum_i(1-\tau_i) - \beta\sum_i n_i}{\beta(1+\beta)}=0;\nonumber\\
	\frac{\partial Q}{\partial p} &=& \frac{\sum_i\tau_i}{p} - \frac{\sum_i(1-\tau_i)}{1-p}=0,\label{eqn:partialp}
	\end{eqnarray}
	where $\digamma(\cdot)$ is the digamma function denoting the logarithmic derivative of a gamma function, i.e.,
	\[
	\digamma(x) = \frac{d}{d x}\log(\Gamma(x)) = \frac{\Gamma'(x)}{\Gamma(x)}.
	\]
	The estimates for the current iteration, i.e., the $(l+1)$-{th}, will be the solutions of the above equations. Note that Equation \eqref{eqn:partialp} depends only on parameter $p$, thus it can be solved directly with explicit representation
	\[
	p^{(l+1)} = \frac{\sum_i \tau_i}{m}.
	\]
	For the other three parameters, we provide numerical results. We employ the {\it nleqslv} R package, which  solves systems of non-linear equations. As a consequence, we can obtain estimated parameters at this iteration $\theta^{(l+1)} = \{\alpha_1^{(l+1)}, \alpha_2^{(l+1)}, \beta^{(l+1)}, p^{(l+1)}\}$.
	\item Plug $\theta^{(l+1)}$ into the $(l+2)$-{th} iteration and repeat the E-M steps until convergence.
\end{enumerate}

Initiating values using the method of moments yields $\alpha_1^{(0)} = 96.14042;\, \alpha_2^{(0)} = 31.54888; \, \beta^{(0)} = 0.01927362;\,  p^{(0)}= 0.6555556$, we can implement the EM algorithm on the chosen dataset mentioned earlier. We note that at a tolerance level of $0.001$, it converge within 75 iterations and the resulting estimates for $\alpha_1, \alpha_2, \beta, p$ are
\[
\alpha_1 = 97.55820446;\, \alpha_2 = 30.14706672;\, \beta = 0.01978072;\,  p = 0.5929959.
\]
Based on these parameters, we add theoretical density to the histogram as shown in Figure \ref{fig:hist}. Visually it appears to be a good fit.
\begin{figure}
	\center
	\includegraphics[scale=0.3]{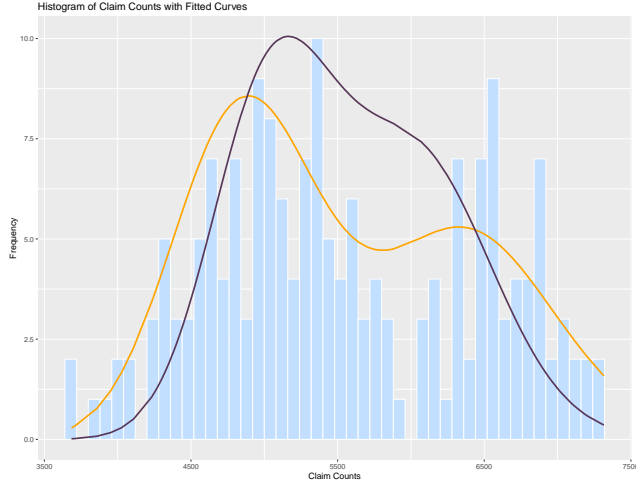}
	\caption{Histogram with Fitted Density}
	\label{fig:hist}
\end{figure}
Note that, however, this result stemmed from initial values being derived from moments. In addition, we found that estimates vary much with different chosen initial values. For instance, if we begin with $\alpha_1^{(0)} = 150; \alpha_2^{(0)}  = 20; \beta^{(0)}  = 0.02,$ convergence happens only within the $308^{th}$ iteration, at the same tolerance level as above. Nevertheless, the corresponding estimates lie closer to the initial values.
\[
\alpha_1 = 141.98069243;  \alpha_2 = 28.02122426; \beta = 0.02773337, p = 0.5723845.
\]
As can be seen in Figure \ref{fig:hist}, we could already tell that most likely it does not provide a better fit than the previous one.

On the other hand, the Kolmogorov-Smirnov, Anderson-Darling and Cram\'{e}r-Von Mises test results are summarised in Table~\ref{table:tests}. 
\begin{table}[t]
	\centering
	\begin{tabular}{|c|c|c|c|}
		\hline
		\thead{Reference Distribution} & \thead{Kolmogorov-Smirnov} & \thead{Anderson-Darling} & \thead{Cram\'{e}r-Von Mises}\\ \hline
		\makecell{NB1(97.558, 0.0198) \\and NB2(127.705, 0.0198)} & 0.6849 & 0.6316 & 0.682 \\ \hline
		\makecell{NB1(141.981, 0.0277) \\and NB2(170.002, 0.0277)} & 0.00162 & 4.347e-06 & 0.001474\\ \hline
	\end{tabular}
\caption{Summary of goodness-of-fit tests}
	\label{table:tests}
\end{table}
In the former case, i.e., when initial values were obtained via method of moments, all $ p $-values are big. That is to say, we do not have enough evidence to reject the null hypothesis that the observed data agree with the proposed model. On the contrary, none of the tests present $ p $-values greater than 0.01 for the second example, i.e., an arbitrarily chosen set of starting values. Therefore, we could conclude the second model does not fit the data well even at 1\% significance level.
\subsubsection*{Claim severity component}
As we said before, we do not have the claim severity data corresponding to the claim counts, so we first simulated them based on the claim counts used earlier. Then we have for claim counts in each period followed by a sequence of claim sizes randomly generated according to \eqref{eqn:clmdist} when $f(\cdot)$ and $g(\cdot)$ take $Exponential(\mu)$ and $Pareto(\delta, \sigma)$ forms, respectively. 

With varying values of the shape parameter in the Pareto part, we considered two separate scenarios when generating claim size data:
\begin{enumerate}[label={(\arabic*)}]
	\item Shape parameter $\delta > 1$, i.e., finite mean;
	\item Shape parameter $\delta \leq 1$, i.e., infinite mean.
\end{enumerate}
For illustration purposes, we adopted $\delta = 2$ and $\delta = 0.3$, respectively, for claim costs simulations. It is well-known that the Pareto distribution in the first scenario has finite mean, whereas it has infinite mean in the second scenario. 

Before proceeding to the EM steps, let us look into the parameter $\nu$. In fact, $\nu$ is connected to the parameters in the claim frequency component. Therefore, once the sample of claim counts is given, we would get an estimated value for $\nu$. Recall from previous sections that,
\begin{equation}
\nu = p + (1-p)\cdot\frac{B(\alpha_1+1, \alpha_2)}{B(\alpha_1, \alpha_2)}\, .
\label{eq:nu}
\end{equation}
If we substitute in \eqref{eq:nu} the estimated values of $p,\, \alpha_1,\, \alpha_2$, we could obtain that the point estimation for $\nu$ would be $\nu = 0.9039196$. $\nu$ could then be treated as known for now so that we can use it in the claim costs simulation. 

Thus, we are ready to implement the steps of EM algorithm for estimating parameters in the claim severity model. 
\begin{enumerate}
	\item We begin with an initially determined parameter values $\vartheta^{(0)} = \{\mu^{(0)}, \delta^{(0)}, \sigma^{(0)}\}$.
	\item E-Step 
	
	For the $(l+1)$-{th} iteration, $l = 0, 1, \ldots$,  we first seek for the expected value of the random value $U_i$, representing the latent variable, conditional on the observations together with the current parameter estimates $\vartheta^{(l)}$, i.e., estimates from the $l$-{th} iteration. Let $ \tau_i $ denote this expectation
	\begin{eqnarray*}
		\tau_i &=& \mathbb{E} [U_i| { Y,\vartheta^{(l)}}] = 1\times \mathbb{P}\{U_i = 1 | Y_i,\vartheta^{(l)}\} + 0 \times \mathbb{P}\{U_i = 0 | Y_i,\vartheta^{(l)}\} \\[0.15cm]
		&=&  \frac{\mathbb{P}\{U_i = 1, Y_i|\vartheta^{(l)}\}}{\mathbb{P}\{Y_i, \vartheta^{(l)}\}}
		=
		\frac{\nu \mu^{(l)} e^{-\mu^{(l)} y_i}}{\nu \mu^{(l)} e^{-\mu^{(l)} y_i} + (1-\nu)\frac{\delta^{(l)} {\sigma^{(l)}}^{\delta^{(l)}}}{(\sigma^{(l)}+y_i)^{\delta^{(l)}+1}}}
	\end{eqnarray*}
	Subsequently, we compute the expectation of its log-likelihood function with respect to the conditional distribution of $U$ given $Y$ under the current estimates $\vartheta^{(l)}$. Considering that we have $m^\ast$ independent observations, we have,
	\begin{eqnarray}
	Q(\vartheta |\vartheta^{(l)}) &= & \mathbb{E}_{U_i | Y_i, \vartheta^{(l)}}[\mathcal{L}(\vartheta | Y_i, U_i)]\label{eqn:clm_expectation}\\
	&=&\sum_{i}^{m^\ast}\tau_i \left[\log \nu + \log \mu^{(l)} -\mu^{(l)} y_i\right]\nonumber\\
	&& \hspace{0.25cm}+ \sum_{i}^{m^\ast} (1-\tau_i) \left[\log (1- \nu) + \log \delta^{(l)} +\delta^{(l)}\log \sigma^{(l)}-(\delta^{(l)}+1) \log(\sigma^{(l)}+y_i)\right]\nonumber\\[1.5mm]
	&=& \log \nu \sum_{i}\tau_i + \log\mu^{(l)} \sum_i \tau_i -\mu^{(l)} \sum_i \tau_i y_i  + \log (1- \nu) \sum_{i} (1-\tau_i) \nonumber\\
	&& + \log\delta^{(l)} \sum_{i} (1-\tau_i) +\delta^{(l)} \log \sigma^{(l)} \sum_{i} (1-\tau_i) - (\delta^{(l)}+1) \sum_{i} (1-\tau_i) \log(\sigma^{(l)}+ y_i).\nonumber
	\end{eqnarray}
	\item M-Step
	
	The maximisation step is set to find the parameter values that maximises \eqref{eqn:clm_expectation}, i.e.,
	\begin{equation}
	\vartheta^{(l+1)} =\operatornamewithlimits{argmax}_{\vartheta} Q(\vartheta |\vartheta^{(l)})
	\label{eqn:clm_max}
	\end{equation}
	In order to do this, we take the gradient of $Q(\vartheta |\vartheta^{(l)})$, equate to zero and solve the system for $\{\mu, \delta, \sigma\}$ simultaneously.
	\begin{eqnarray}
	\frac{\partial Q}{\partial \mu} &=& \frac{\sum_i \tau_i}{\mu}-\sum_i \tau_i y_i =0;\nonumber\\
	\frac{\partial Q}{\partial \delta} &=& \frac{m^\ast-\sum_i\tau_i}{\delta} + m^\ast\log\sigma -\log \sigma\sum_i \tau_i -\sum_i (1-\tau_i) \log(\sigma+y_i) =0 ;\nonumber\\
	\frac{\partial Q}{\partial \sigma} &=& \frac{(m^\ast-\sum_i\tau_i)\delta}{\sigma} - (\delta+1)\sum_i \frac{1-\tau_i}{\sigma+y_i}=0 \, . \label{eqn:clm_partialp}\\
	\frac{\partial Q}{\partial \nu} &=& \frac{\sum_i\tau_i}{\nu} - \frac{\sum_i(1-\tau_i)}{1-\nu}.\nonumber
	\end{eqnarray}

	First note that implementing the EM algorithm to find the estimate $\hat{\nu}$ is very similar to that applied to $p$ in Subsection~\ref{sec:emfreq}. Recall that the partial derivative equation with respect to $p$ in that subsection is independent from all other variables. A similar situation exists here for $\nu$. Note that it does not affect the remainder of the equations. As so, we can find
	\begin{equation}
	\hat{\nu}^{(l+1)} = \frac{\sum_i \tau_i}{m^{\ast}}.
	\end{equation}

The estimates for the subsequent iteration $(l+1)$-{th} are the solutions to the above equations.\\
	It is also obvious that $\mu$ is independent from other parameters and could be solved directly with an explicit representation
	\[
	\mu^{(l+1)} = \frac{\sum_i \tau_i}{\sum_i \tau_i y_i}.
	\]
	For the other two parameters, we could only solve numerically. Again we apply the {\it nleqslv} package in R. As a consequence, we can derive estimated parameters at this iteration $\vartheta^{(l+1)} = \{\mu^{(l+1)}, \delta^{(l+1)}, \sigma^{(l+1)}\}$.
	
	\item Plug $\vartheta^{(l+1)}$ into the $(l+2)$-{th} iteration and repeat the above steps until convergence.
\end{enumerate}

Then, we implemented the algorithm and were able to achieve the desired parameters within reasonable amount of iterations. In Table~\ref{tab:est1} and Table~\ref{tab:est2} we show our estimates against the predefined parameters in the simulation scenario with $\delta > 1$ and $\delta \leq 1$, respectively. This verifies the effectiveness of our algorithm. 

\begin{table}[h]
	\centering
	\begin{tabular}{|l|c|c|c|c|}
		\hline
		{\bf Parameters} & $\mu$ & $\delta$ & $\sigma$ & $\nu$\\ \hline\hline
		Predefined & 1 & 2 & 1 & 0.9039196\\ \hline
		Initial value & 1.5 & 2.5 & 0.5 & 0.9\\ \hline
		After 2467 iterations (at convergence) & 0.9925845 & 2.219456 & 1.159886 & 0.8343595\\
		\hline
	\end{tabular}
	\caption{Performance of estimation using EM algorithm on simulated claim severities ($\delta > 1$)}
	\label{tab:est1}
\end{table}
\begin{table}[h]
	\centering
	\begin{tabular}{|l|c|c|c|c|}
		\hline
		{\bf Parameters} & $\mu$ & $\delta$ & $\sigma$ & $\nu$\\ \hline\hline
		Predefined & 1 & 0.3 & 0.5 & 0.9039196\\ \hline
		Initial value & 1.5 & 0.5 & 0.2 & 0.9\\ \hline
		After 1330 iterations (at convergence) & 1.001311 & 0.2973774 & 0.5146515 & 0.9050024\\
		\hline
	\end{tabular}
	\caption{Performance of estimation using EM algorithm on simulated claim severities ($\delta \leq 1$)}
	\label{tab:est2}
\end{table}
The tables have demonstrated satisfactory results using the EM algorithm for mixed distribution of the claim severity component. We would like to point out, however, that the algorithm is highly sensitive to the chosen initial values. Some starting values could even lead to no solutions to the non-linear equations to be solved in the algorithm. But since we knew the true parameter values here, we were able to begin with a relatively reasonable initial 'guess' of the parameter values. In the case of applying it to real data, similar to estimating the claim frequency component explained in the previous section, the method of moments could be adopted in the beginning. Figure~\ref{fig:est} shows how the parameter estimates evolved over iterations in the EM algorithm.
\begin{figure}
	\center
	\includegraphics[scale=0.45]{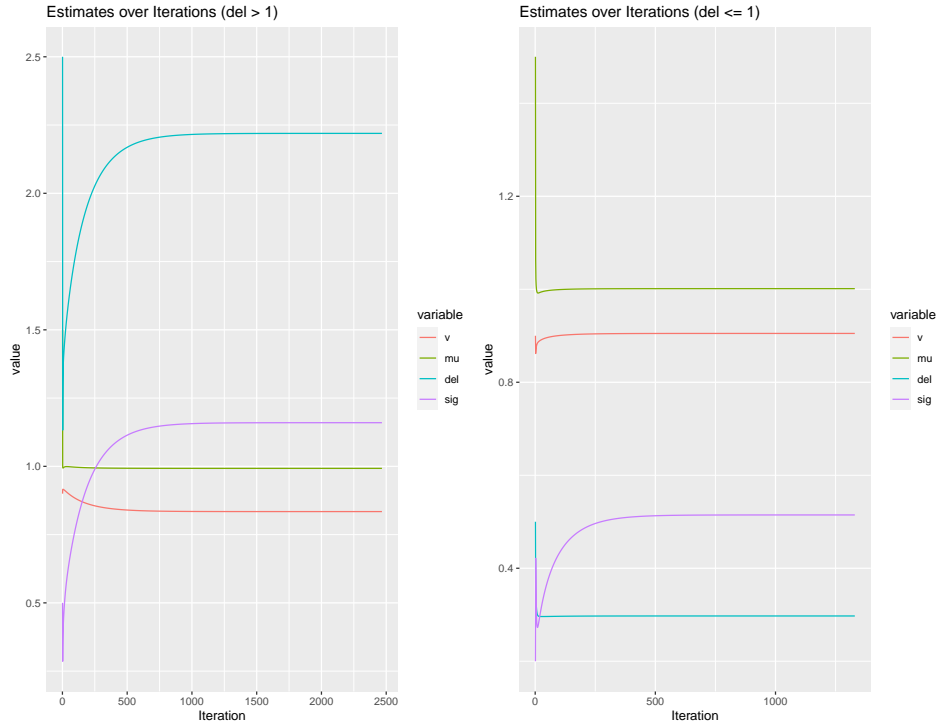}
	\caption{Parameter Estimates under EM Algorithm}
	\label{fig:est}
\end{figure}

Similarly, we can visualize how the distribution is fitted to the claims from the histogram with the fitted density. On the top plot in Figure~\ref{fig:fitclaims1}, we not only show the fit of the mixed distribution, but also the fits using an Exponential as well as a Pareto distribution to the simulated data without extreme claims, i.e., $\delta = 2$. 

However, the goodness-of-fit is not what we would like to address here, because the claim costs data were simulated. We are more interested in the tail behaviours using different models. The tail behaviour graphs are plotted under log-log scale in both Figure~\ref{fig:fitclaims1} and Figure~\ref{fig:fitclaims2}, from which we can clearly see that our mixed model provides the heaviest fitted tails regardless whether the claim data contain extreme values or not. It implies that the use of mixed distribution to fit the claims leads to a better capture of information in the tails. In other words, the mixed distribution is more sensitive in detecting large claims, which reflects our motivation for this work.

\begin{figure}
	\center
	\includegraphics[scale=0.45]{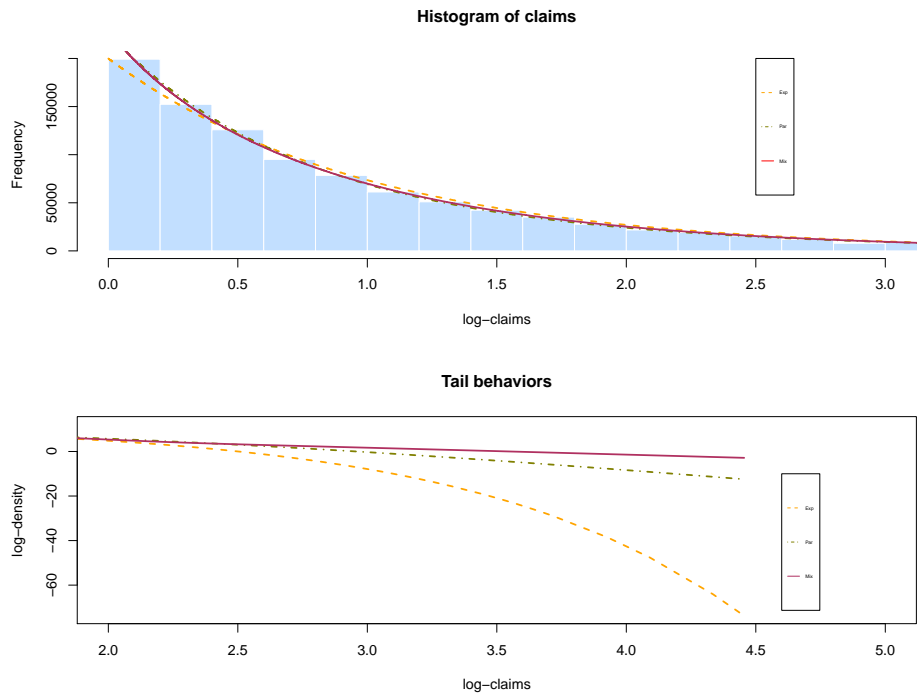}
	\caption{Histogram and Tail Behaviours ($\delta > 1$)}
	\label{fig:fitclaims1}
\end{figure}
\begin{figure}
	\center
	\includegraphics[scale=0.3]{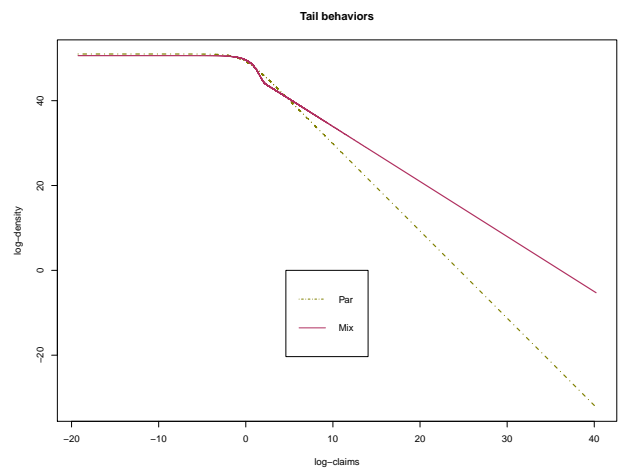}
	\caption{Tail Behaviours ($\delta \leq 1$)}
	\label{fig:fitclaims2}
\end{figure}
\subsection{Premium Calculation}
With the estimated parameters, we can then calculate the Bayesian premiums according to the theoretical results discussed in Section \ref{sec:Bayesian}. To clarify, by Bayesian premium, we refer to the mean of the posterior distributions. Under our setting, the posterior for both the frequency and severity are mixed distributions. By choosing mixed conjugate priors that are in the Exponential family, we were able to show that the posterior distribution is also a mixture with updated parameters which are dynamically adjusted by incoming claims. 
As a reminder, we proposed the use of a mixture of two Negative Binomial distributions for the frequency, and a mixture of an Exponential and a Pareto distribution for the severity. To illustrate our proposed model, we also calculated Bayesian premiums (posterior means) using other common models, i.e., when claim counts are fitted by Negative Binomial, whereas claim severities are fitted by Exponential or Pareto, as we did in the severity data fitting. All models under comparison are listed in Table~\ref{tab:models}. 
\begin{table}[h]
	\centering
	\begin{tabular}{|c|c|c|c|}
		\hline
		{\bf Scenario} & {\bf Frequency} & {\bf Severity} & {\bf Model} \\\hline\hline
		\multirow{3}{*}{1} & Negative Binomial & Exponential & 1a\\\cline{2-4}
		& Negative Binomial & Pareto ($\delta > 1$) & 1b\\\cline{2-4}
		& Mixed Negative Binomial & Mixed Exponential and Pareto ($\delta > 1$) & 1c\\\hline
		\multirow{2}{*}{2} & Negative Binomial & Pareto ($\delta \leq 1$) & 2a \\\cline{2-4}
		& Mixed Negative Binomial & Mixed Exponential and Pareto ($\delta \leq 1$) & 2b\\\hline
	\end{tabular}
	\caption{Parametric models used for premium calculations}
	\label{tab:models}
\end{table}

We can thus compute all their corresponding posterior means for each period based on the observed claim counts and simulated costs over time. Table~\ref{tab:premLight} and Table~\ref{tab:premHeavy} display the posterior means for the first 10 periods under the two separate cases distinguished by the values of the parameter $\delta$. Note that when the shape parameter $\delta \leq 1$ originating from the Pareto distribution, the claims can explode to infinity. Hence, very large values can be seen from Table~\ref{tab:premHeavy}, and the claim counts are omitted because the posterior means are mostly driven by the claim costs. In addition, fitting an Exponential distribution to these claims would result in the parameter being virtually zero. That is also why there is no 'Exponential' model in Table~\ref{tab:models} or 'Exp' column in Table~\ref{tab:premHeavy} either.

\begin{table}[h]
	\centering
	\begin{tabular}{|c|c|c|c|c|c|}
		\hline
		{\bf Period} & {\bf Exp (1a)} & {\bf Par (1b)} & {\bf Mix (1c)} & {\bf Counts} & {\bf Costs}\\\hline\hline
		1 & 4967.668 & 5021.675 & 5017.420 & 4964 & 5017.763\\\hline
		2 & 4684.657 & 4708.717 & 4708.197 & 4400 & 4393.788\\\hline
		3 & 4632.124 & 4647.323 & 4647.190 & 4527 & 4524.171\\\hline
		4 & 4646.491 & 4667.480 & 4667.337 & 4690 & 4728.054\\\hline
		5 & 4649.527 & 4660.096 & 4659.971 & 4662 & 4630.573\\\hline
		6 & 4612.599 & 4640.139 & 4640.116 & 4428 & 4540.106\\\hline
		7 & 4652.571 & 4683.893 & 4683.802 & 4893 & 4946.695\\\hline
		8 & 4672.193 & 4711.397 & 4711.288 & 4810 & 4904.065\\\hline
		9 & 4626.397 & 4659.045 & 4659.011 & 4260 & 4239.913\\\hline
		10 & 4566.472 & 4600.162 & 4600.208 & 4027 & 4069.803\\\hline
	\end{tabular}
	\caption{Posterior Means under different models ($\delta > 1$)}
	\label{tab:premLight}
\end{table}
\begin{table}[h]
	\centering
	\begin{tabular}{|c|c|c|c|c|}
		\hline
		{\bf Period} & {\bf Par (2a)} & {\bf Mix (2b)} & {\bf Costs}\\\hline\hline
		1 & 3.258448e+09 & 3.256148e+09 & 3.255862e+09\\\hline
		2 & 1.717028e+14 & 1.716975e+14 & 3.431870e+14\\\hline
		3 & 1.145418e+14 & 1.145447e+14 & 2.808137e+11\\\hline
		4 & 8.589659e+13 & 8.589740e+13 & 1.786957e+09\\\hline
		5 & 6.871266e+13 & 6.871307e+13 & 4.558560e+07\\\hline
		6 & 5.726841e+13 & 5.726965e+13 & 5.845177e+10\\\hline
		7 & 4.908520e+13 & 4.908535e+13 & 7.294295e+08\\\hline
		8 & 4.294886e+13 & 4.294869e+13 & 3.949388e+09\\\hline
		9 & 3.817744e+13 & 3.817783e+13 & 9.165001e+09\\\hline
		10 & 3.457031e+13 & 3.457120e+13 & 2.107256e+12\\\hline
	\end{tabular}
	\caption{Posterior Means under different models ($\delta \leq 1$)}
	\label{tab:premHeavy}
\end{table}

For a better understanding of the differences among different models, we present in graphs the changes of premiums overtime. Figure~\ref{fig:premLight} demonstrates the first scenario when there are no extreme claims. Figure~\ref{fig:premHeavy} is the graphical representation of the quantities in Table~\ref{tab:premHeavy}.
\begin{figure}
	\center
	\includegraphics[scale=0.45]{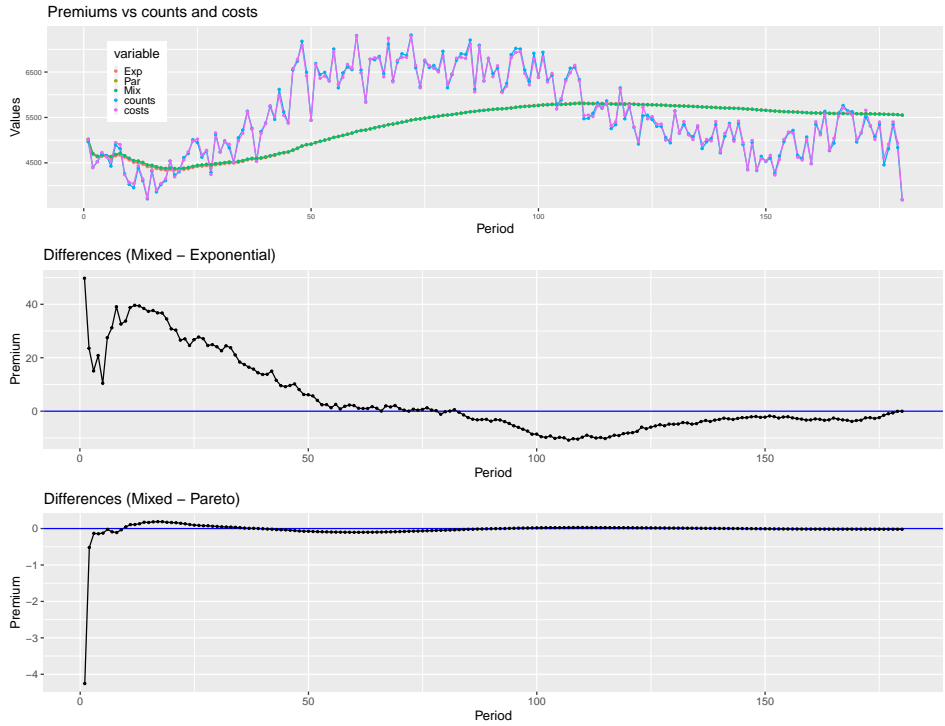}
	\caption{Bayesian Premiums ($\delta > 1$)}
	\label{fig:premLight}
\end{figure}
In Figure~\ref{fig:premLight}, the top graph contains claim counts and costs per period as well as the posterior means obtained from Model 1a, Model 1b, and Model 1c, respectively. In general, the differences among the three models are not obvious in the first plot. So the bottom two figures show the exact differences between the pairs we are interested in, i.e., 1a vs 1c and 1b vs 1c. 

Comparing our mixed model (1c) with the Exponential only one (1a), we can see significant discrepancies over time. Matching them with the trend of the claim counts/costs, we can conclude that the mixed model provides higher premiums when there is an increasing trend in claims (e.g. from period 10 to 50), while it immediately drops below the premium suggested by the Exponential model as soon as a decreasing trend is observed (e.g. from period 100 to 150). That has reflected a quick response in the premium adjustment of the mixed model. 

There are no such evident difference between the mixed model (1c) and the Pareto only model (1b) overall, as can be seen in the graph at the bottom. Nevertheless, the mixed model gives slightly higher premiums when there is a surge in the claims around period 10. Another main difference is that the mixed model serves lower premiums at the very beginning compared to the Pareto model. It matches with our original intention that no punishment should be there for no claims.

To sum up, these results have shown support for our proposal of the mixed model. On the one hand, we would like to keep the promise of not increasing premium charges until we have evidence of risks. On the other hand, the premium adjustment to actual evidence is sensitive and spontaneous.
\begin{figure}
	\center
	\includegraphics[scale=0.45]{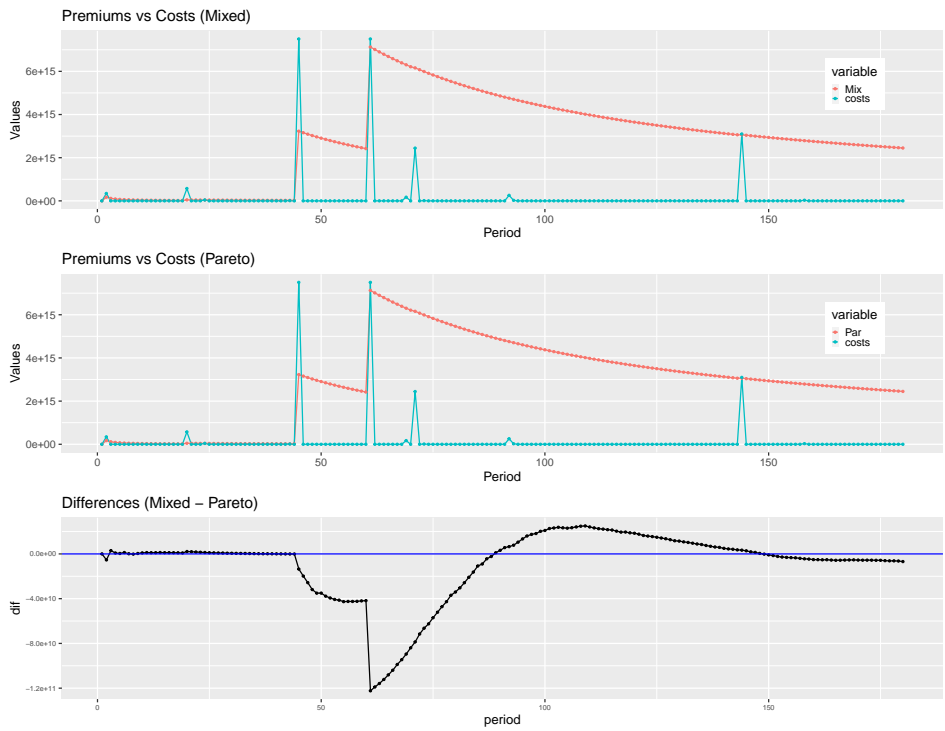}
	\caption{Bayesian Premiums ($\delta \leq 1$)}
	\label{fig:premHeavy}
\end{figure}
Figure~\ref{fig:premHeavy} exhibits the premiums (posterior means) obtained from Model 2a and 2b, respectively. They have demonstrated very similar trends overall according to the claim costs behaviours. When there is a big jump in the claim severity, e.g. around period 47, both models are able to provide a surge in the premiums instantaneously. However, due to the large scale of claim severities, we could not see clearly the difference between the two models from the top two graphs. The third plot was thus generated at the bottom. As can be seen here, the mixed model in this case actually advocates lower premiums when the claim costs are extreme, compared to the Pareto model. It could be an effect stemming from the mixture with the Exponential distribution. However, from the practical point of view, the premiums would probably be capped at a level anyway when the scale is that large. 

We have been discussing the mean values only. Since we have derived the posterior distributions, it would be interesting to visualize other statistics, such as the inter-percentile ranges (IPR). Note that the IPR is a reflection of the variance of the posterior distribution, and is also referred to as the $(1-\alpha)\%$ credibility interval under the Bayesian context. We will be using both terminologies interchangeably in the sequel. Here we have sketched the 90\% credibility interval (5-95 IPR) for the first scenario, i.e., for Model 1a, 1b and 1c in Figure~\ref{fig:IPR}. That includes the 90\% possible values in the center for the risks. 
\begin{figure}
	\center
	\includegraphics[scale=0.3]{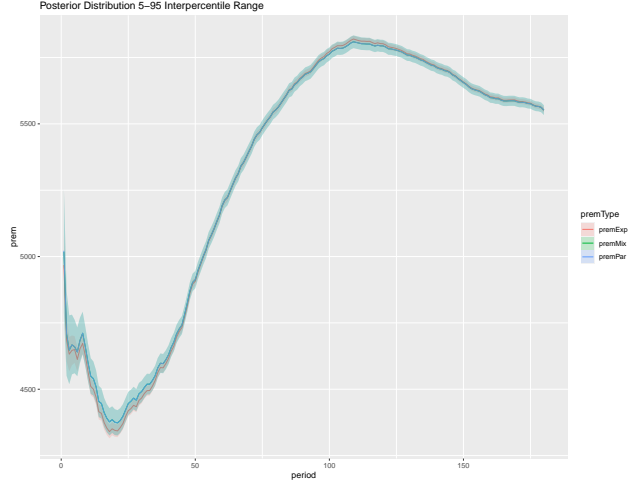}
	\caption{90\% Credibility Interval ($\delta > 1$)}
	\label{fig:IPR}
\end{figure}
As shown in the graph, all three models have premium values lie very close to each other, with Exponential (red) to be more distinct from the other two. Overall, the credibility interval differ more at the very beginning possibly due to insufficient data collected. Exponential has essentially the narrowest credibility interval amongst the three models. For comparison between the actual values of the 90\% credibility interval between model 1b and 1c, see Figure~\ref{fig:LVar}. Except for a significantly small credibility interval at the very beginning, they have exhibited no difference in the variation. This has, however, offered us with one more reason to choose the mixed model: fewer uncertainties are preferred when assessing risks. 

\begin{figure}
	\center
	\includegraphics[scale=0.4]{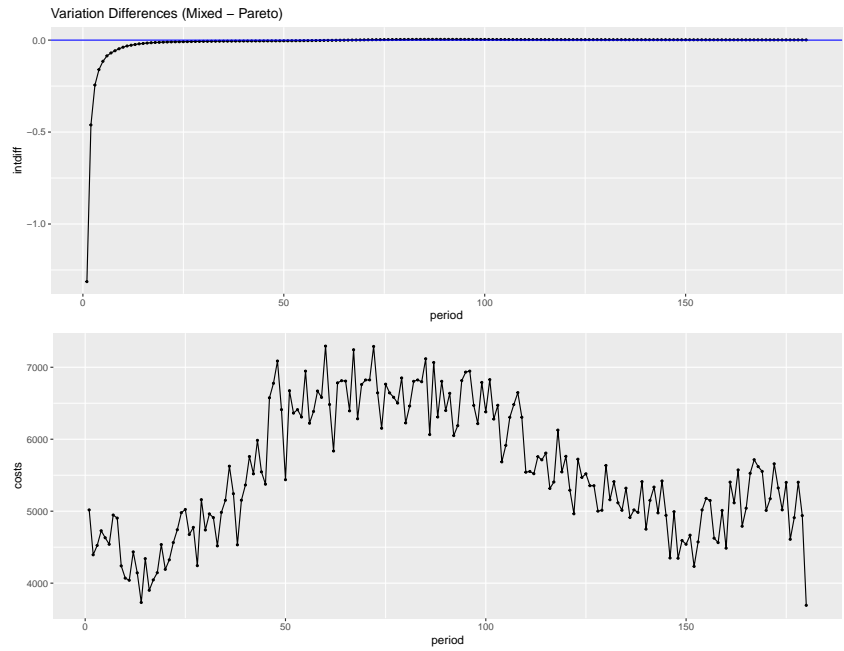}
	\caption{Variation Differences ($\delta > 1$)}
	\label{fig:LVar}
\end{figure}
In the extreme case, we drew a graph of comparison in a similar way. As a result, we can observe a smaller variation at the beginning and also when the claims explode. Again, it shows support for the mixed model from the smaller variance perspective, especially when other premium principles (e.g. the variance premium principle) are used instead of the pure premium principle adopted in our setup.
\begin{figure}
	\center
	\includegraphics[scale=0.4]{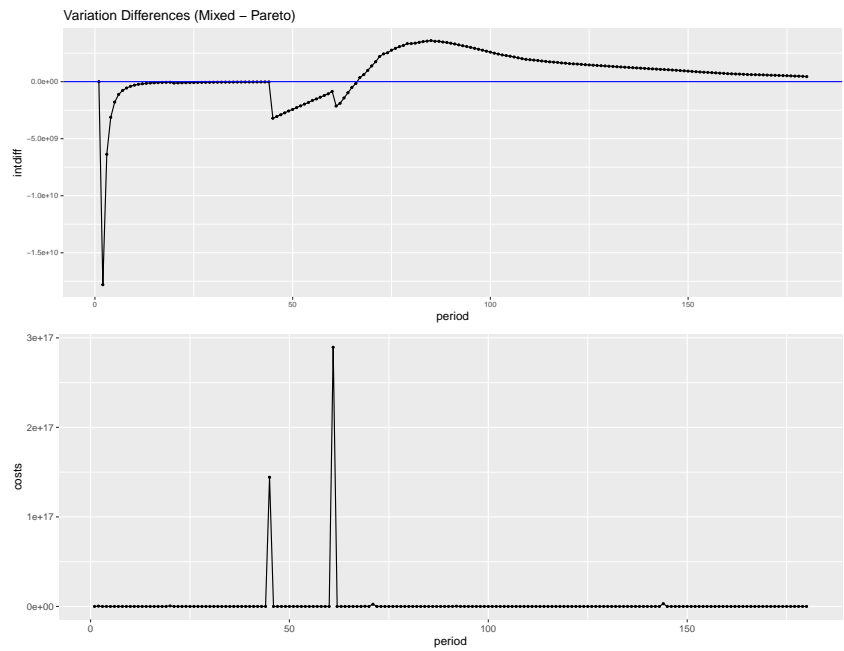}
	\caption{Variation Differences ($\delta \leq 1$)}
	\label{fig:HVar}
\end{figure}
\subsection{Findings}
 Based on the numerical results and comparison with other commonly used models, we can make a few summary statements. On the one hand, we proposed a kind of zero-inflated model that can be more practically employed by actuaries. According to the dataset considered here, the estimated probability of zeros per period is approximately 60\%, which is larger than traditional models. That implies a lower initial premiums which could facilitate the sales of policies. On the other hand, even though our model assesses the risks by allowing for more excess of zeros, we have set up a mechanism under the Bayesian framework such that the measurement of risks can be adjusted in a timely manner to the observed extreme events. Premiums suggested by the model will dynamically and automatically modify to incorporate the potential risks involved. In addition, the uncertainties involved in estimating risks are relatively small in our model. 

Therefore, our model could be useful in practice for its ability of sensitive detection of risks and efficient adjustment to risk measurement. Early detection of risky behaviours in policyholders would be most valuable for an insurer. This particularly applies to innovative policies covering e.g., autonomous cars and COVID-19, where initial information is very little or missing and discernment to potential risks is needed.
\section{Discussions on global likelihood} \label{sec:global}
In the previous sections, we were considering the independence between the claim frequency and severity components. It would be opportune to investigate the effect of dependence between the two 
quantities. One possible consideration has been done in \cite{cheung2019note} by adopting a conjugate bivariate prior distribution. Here, we will briefly introduce the idea of implementing a global 
likelihood in this section which takes into account {the two} components simultaneously for parameter estimations. 

In general, a random sample will be composed by a sequence of $m$ independent pairs of dependent observations $(N_i, \overrightarrow{Y_i})$, where $N_i$ represents the number of claims in the $i$-{th} period and $\overrightarrow{Y_i} = (Y_{i1}, \ldots, Y_{in_i})$ is the corresponding sequence of claim severities. It is clear that if we consider in general that both claim counts and severities may bring information about each stream, foreseable or unforeseable, we should consider that stochastic dependence between $N_i$ and corresponding  $\overrightarrow{Y_i}$ is present. However, they are conditionally independent, given $\Lambda=\Lambda_i\, $ ($i=1,2$) that is, for each individual stream we can consider the classical assumption of independence between claim counts and severities. 

Observed values are represented by corresponding lower case letters $(n_i, \overrightarrow{y_i})$. We also note that in each pair the dimension of vector $\overrightarrow{y_i}$ depend on the observed $n_i$. Now, we could write a likelihood function considering the joint random vector $(N, \overrightarrow{Y})$ assuming we have $m$ groups of observations and conditional independence of each $Y_{ij}$, $j = 1, \ldots, n_i$, for a given $i$:
\begin{eqnarray*}
	L(\vartheta | \mathbf{n, \overrightarrow{y}}) = \prod_{i=1}^m f_{N, \overrightarrow{Y}}(n_i, y_{i1}, \ldots, y_{in_i}) &=& \prod_{i = 1}^m f_{\overrightarrow{Y}|N}(\overrightarrow{y_i}| n_i) \mathbb{P}(n_i) = \prod_{i=1}^m \left\{\prod_{j=1}^{n_i} f_{\overrightarrow{Y}|N}(y_{ij}| n_i)\right\} \mathbb{P}(n_i)\,,
\end{eqnarray*}
where $\mathbb{P}(n_i)=\mathbb{P}(N=n_i)$, for simplification. 

For our model, we can write the corresponding global log-likelihood function, denoted as $\mathcal{L}(\vartheta | \mathbf{n, \overrightarrow{y}})$. If  we let claim counts conform to Lemma~\ref{l:mix_nb} and claim severities follow \eqref{eqn:clmdist}, i.e.~$f(\cdot)$ and $g(\cdot)$ are respectively Exponential and Pareto densities, we have
{\small\begin{eqnarray*}
		\mathcal{L}(\vartheta | \mathbf{n, \overrightarrow{y}}) &=& \sum_{i=1}^m \log \mathbb{P}(n_i) + \sum_{i=1}^m \sum_{j=1}^{n_i} \log f_{\overrightarrow{Y}|N}(\overrightarrow{y_i}| n_i)\\
		&=& \sum_{i=1}^m \log\left\{p\binom{n+\alpha_1-1}{n}\left(\frac{\beta}{\beta + 1}\right)^{\alpha_1}\left(\frac{1}{\beta + 1}\right)^n \right.\\
		& & \;\;\; \;\;\; 
		\left. + (1-p)\binom{n+\alpha_1+\alpha_2-1}{n}\left(\frac{\beta}{\beta + 1}\right)^{\alpha_1+\alpha_2}\left(\frac{1}{\beta + 1}\right)^n\right\}
		\\
	    && 
		 \;\;\; \;\;\; \;\;\; \;\;\; +
		 \sum_{i=1}^m \sum_{j=1}^{n_i} \log \left\{ \nu \mu e^{-\mu y_{ij}} + (1-\nu) \frac{\delta \sigma^\delta}{(\sigma+ y_{ij})^{\delta + 1}}\right\}.
\end{eqnarray*}}

Instead of building a likelihood function based on a sample of observations of the pair $(N, \overrightarrow{Y})$, we can build it using inter-arrival time and corresponding severity, where each observation is a bivariate pair, denoted as $(T_j, Y_j)$, of dependent random variables, in general. Let's denote the density and the distribution function of $T_j$ by $\phi(\cdot)$ and $\Phi(\cdot)$, respectively. 

For finding the distribution of $T_j$ recall that, conditional on $\Lambda =\lambda$, $\{N(t)\}$ is a  Poisson process with intensity $\lambda$. Then, given $\Lambda =\lambda$, conditional interarrival time $T_j|\Lambda =\lambda$ are exponentially distributed, with mean $\lambda^{-1}$. Whenever the Poisson parameter is an outcome of a random variable $\Lambda$ and it is distributed as \eqref{eqn:gammamix}, we have that the unconditional distribution of $T_j$ is a mixture of two Pareto distributions whose density is given by 
	\begin{equation}
	\label{eq:mixpareto}
	\phi(t_j)=p\,\frac{\alpha_1\beta^{\alpha_1}}{(\beta+t_j)^{\alpha_1+1}}
	+ (1-p)\, \frac{(\alpha_1+\alpha_2)\beta^{\alpha_1+\alpha_2}}{(\beta+t_j)^{\alpha_1\alpha_2+1}} \,.
	\end{equation}
On the other hand, the 
 density function of $Y_j$, conditional on a given $\Xi=\xi$, is given by the mixture, see~\eqref{eq:dfYxi},
	\begin{equation}
	\label{eq:mix_fg}
	h_\xi(y_j)= \xi f(y_j)+(1-\xi)g(y_j)\, ,
	\end{equation}
	and the unconditional density is given by \eqref{eqn:clmdist}. 
	 In particular, we assumed in our model that $f\sim exp(\mu)$, $g\sim Pareto(\delta,\sigma)$. 

From Remark~\ref{rem:indep}, we know that
for  $\Xi=\xi$ fixed, $T_j$ and $Y_j$ are independent. A random sample of observations are made of independent pairs $(Y_j, T_j)$, each $j$ of dependent $Y_i$ and $T_i$. Each pair, $j$, has density function
\begin{eqnarray}
	\label{eq:Ltheta}
	\phi(y_i,t_i) & = & \int_{\xi}h_\xi(y_i)\phi(t_i)dA(\xi) \\
	& = & (1-p) \int_{\xi\neq 1}h_\xi(y_i)\phi(t_i)dA(\xi) + p\, h_1(y_i)
	\frac{\alpha_1\beta^{\alpha_1}}{(\beta+t_i)^{\alpha_1+1}} \,  \nonumber
\end{eqnarray}
 where $A(\xi)$ stands for the distribution function of $\Xi$. 
 
 Note that that we separated the situations $\xi\neq 1$ and $\xi=1$. Also, the equivalent events  are$\{\Xi=1\}\Leftrightarrow \{ \Lambda^{(2)}=0\}$, implying that $Pr\{\Xi=1\}=Pr \{ \Lambda^{(2)}=0\}=p$.
Although we  consider the unconditional distribution of the arrival time (unconditional of $\Lambda$), we consider the conditional distribution for the individual claim size $Y$, given $\Xi=\xi$. 
$\Lambda^{(1)}$, given as a proportion of $\Lambda$, remains random and it  means that the distribution of $\Lambda^{(1)}$ is a scaled distribution of that of $\Lambda$, or \textit{vice versa}. Looking at the density $h_\xi(y_i)$ we see that the split rate $\xi$ gives the weight for the foreseeable stream claim amount, and looking at the density $\phi(y_j)$ we could think that the probability $p$ gives a similar meaning regarding the claim count stream. It does not seem to be the case, as it is not clear that the second part in the mixture represents the unforeseeable only.  

Returning to the joint density~\eqref{eq:Ltheta}, we build the likelihood function over $m^\ast$ pairs of observations ($m^\ast$ may be different from sample size $m$ from random vector $(N, \overrightarrow{Y})$ above)
\begin{eqnarray*}
L(\vartheta) &= &\prod_{i=1}^{m^\ast}\Phi(y_i,t_i)=
\prod_{i=1}^{m^\ast} \int_{\xi}h_\xi(y_i)\phi(t_i)dA(\xi) \\
& = & \prod_{i=1}^{m^\ast} \left[
(1-p)\int_{\xi\neq 1} \left(\xi f(y_i)+(1-\xi)g(y_i)\right) \phi(t_i)Beta(\xi;\alpha_1,\alpha_2)d\xi +p\, f(y_i)\, \frac{\alpha_1\beta^{\alpha_1}}{(\beta+t_i)^{\alpha_1+1}}
\right] ,
\end{eqnarray*}
where $Beta(\xi;\alpha_1,\alpha_2)$ denotes the Beta density function of $\Xi$, given $\Xi\neq 1$, and $\frac{\alpha_1\beta{\alpha_1}}{(\beta+t_i)^{\alpha_1+1}}$ is the density $\phi(t_j)$ when $\xi=1$.
Note that we cannot interchange the product ($\prod_{i = 1}^{m^\ast}$) and the integral. 

\section{Conclusions}
To conclude, on the model and the estimation procedure, this {article} has carried out a development and subsequent parameter estimation for the so-called unforeseeable risks discussed in 
\cite{li2015risk} for the claim counts arrival process. Precisely, for these risks, a probability $p$ has been assigned at mass point $\{0\}$ so that their corresponding counting process is 
distinguished from the classical one. Since we could only observe the entire claim counts from an insurance portfolio, this missing information could be estimated using the EM algorithm. Under certain 
assumptions for the distribution of heterogeneities within the portfolio (which was denoted as $\Lambda$ in this work), we could derive the randomness of the total claim counts given a fixed period to 
be a Negative Binomial mixture distribution. Thus, the likelihood function could be presented explicitly via the EM algorithm.  

We introduced a similar modelling to account the unforeseeable effect in the claim severity as well, by choosing a particular mixture distribution. This ended up with the existence of some common parameters inserted in the severity distribution to those of the claim counts one. This intends to show somehow the extra and common effect of an unforeseeable risk effect in both number of claims and severity. 

When implemented on a data set of claim counts for the third party liability insurance portfolio, from a Portuguese automobile insurer, the resulting estimates show high sensitivity to the chosen initial values. Hence, we employed starting values which were computed via method of moments. {All the Kolmogorov-Smirnov, Anderson-Darling and Cram\'{e}r-Von Mises tests} suggest a good fit on the observed data.

Concerning the claim severity component, due to lack of real data, we randomly generated data based on proposed mixed distribution and separated into two cases depending on whether the mean values are finite or infinite. We verified the EM algorithm by comparing with the true parameter values used in the simulation procedure.
	
When comparing with other traditional models, our model stands out for offering lower starting premia, faster adjustment to claim changes, and baring smaller variation in the resulting values. These are valuable features in practice and would be especially useful for pricing innovative policies covering such as autonomous cars or COVID-19. 

In our model we could consider the calculation of Bayesian and credibility premia separately for claim counts and severities, due to the independence situation between these two random variables. However, it would be feasible to compute a premium where we estimate parameters $ (\alpha_1,\alpha_2, \beta, \mu, \delta, \sigma, p)$ altogether, using a likelihood where $N$ and $Y$ are jointly distributed. We will refer to it as the global likelihood function and it will serve as a direct extension of the current work (as discussed).

Estimating a premium taking into account both foreseeable and unforeseeable risks not only makes the premium fairer, but also helps not underestimating portfolio's ruin probabilities. 

\section*{Acknowledgements}
{Authors thank Fundaci\'on MAPFRE, through 2016 Research Grants Ignacio H. de Larramendi, for the financial support given to the project [{\it ``TAPAS (Technology Advancement on Pricing Auto inSurance)"}]
	
Main research was done while Author $(^1)$ was staying on location at Research Centre CEMAPRE, ISEG, Universidade de Lisboa, Portugal. 
{Authors $(^1)$,$(^3)$ gratefully acknowledge support from Project CEMAPRE/REM--UIDB/05069/2020 - financed by FCT/MCTES (Funda\c c\~ao para a Ci\^encia e a Tecnologia/Portuguese Foundation for Science and Technology) through national funds.

\thanks{Authors $(^1)$,$(^4)$ gratefully acknowledge support by the LABEX MILYON (ANR-10-LABX-0070) of Universit\'{e} de Lyon, within the program ``Investissements d'Avenir" (ANR-11-IDEX- 0007) operated by the French National Research Agency (ANR)} .
}
\nocite{couvreur1997algorithm} \nocite{dempster1977maximum} \nocite{buhlmann2006course}
\nocite{redner1984mixture}

\bibliography{refs}

\appendix
\section{Proofs of Lemmas within the paper}

\begin{proof} Proof of Lemma \ref{l:mix_nb}.
Taking the moment generating function of $N$, as function of $\rho$,  yields

\begin{eqnarray*}
M_N(\rho) &=& \mathbb{E}\left[e^{\rho\,(N^{(1)}+N^{(2)})}\right] = M_{N^{(1)}}(\rho)M_{N^{(2)}}(\rho)\\
&=&\left(\frac{1-\frac{1}{1+\beta}}{1-e^\rho \frac{1}{1+\beta}}\right)^{\alpha_1}\left[p+(1-p)\left(\frac{1-\frac{1}{1+\beta}}{1-e^\rho \frac{1}{1+\beta}}\right)^{\alpha_2}\right]\\
&=&p\left(\frac{1-\frac{1}{1+\beta}}{1-e^\rho \frac{1}{1+\beta}}\right)^{\alpha_1}+(1-p)\left(\frac{1-\frac{1}{1+\beta}}{1-e^\rho \frac{1}{1+\beta}}\right)^{\alpha_1+\alpha_2}.
\end{eqnarray*}
We know that this corresponds to a weighted average of two Negative Binomial distributions, with weights $p$ and $1-p$ respectively. \hfill{}
\end{proof}\\

\begin{proof} Proof of Lemma \ref{l:dominik}.
Recall that given $\Lambda = \lambda$,
$
\mathbb{P}\{N=n |\Lambda =\lambda\} = {e^{-\lambda}\lambda^n}/{n!}
$.
Now we integrate over $\Lambda$ whose pdf, denoted as $\pi(.)$, can be written as
\begin{equation}
\label{eqn:gammamix}
\pi(\lambda) = p\cdot\frac{\lambda^{\alpha_1-1}\beta^{\alpha_1} e^{-\beta \lambda}}{\Gamma(\alpha_1)} + (1-p)\cdot\frac{\lambda^{(\alpha_1+\alpha_2)-1}\beta^{\alpha_1+\alpha_2} e^{-\beta \lambda}}{\Gamma(\alpha_1+\alpha_2)}.
\end{equation}
Hence,
\begin{eqnarray*}
\mathbb{P}\{N=n\} &=& \int_\Lambda \frac{e^{-\lambda}\lambda^n}{n!} \left(p\cdot\frac{\lambda^{\alpha_1-1}\beta^{\alpha_1} e^{-\beta \lambda}}{\Gamma(\alpha_1)} + (1-p)\cdot\frac{\lambda^{(\alpha_1+\alpha_2)-1}\beta^{\alpha_1+\alpha_2} e^{-\beta \lambda}}{\Gamma(\alpha_1+\alpha_2)}\right)\,\rm d\lambda\\
&=& p \cdot\int_\Lambda \frac{e^{-\lambda}\lambda^n}{n!}\frac{\lambda^{\alpha_1-1}\beta^{\alpha_1} e^{-\beta \lambda}}{\Gamma(\alpha_1)}\,\rm d\lambda + (1-p)\cdot\int_\Lambda \frac{e^{-\lambda}\lambda^n}{n!}\frac{\lambda^{(\alpha_1+\alpha_2)-1}\beta^{\alpha_1+\alpha_2} e^{-\beta \lambda}}{\Gamma(\alpha_1+\alpha_2)}\,\rm d\lambda\\
&=&p\binom{n+\alpha_1-1}{n}\left(\frac{\beta}{\beta + 1}\right)^{\alpha_1}\left(\frac{1}{\beta + 1}\right)^n\\
&& + (1-p)\binom{n+\alpha_1+\alpha_2-1}{n}\left(\frac{\beta}{\beta + 1}\right)^{\alpha_1+\alpha_2}\left(\frac{1}{\beta + 1}\right)^n \,.
\end{eqnarray*}
This distribution coincides with the one shown in Lemma~\ref{l:mix_nb}. \hfill{}
\end{proof}\\

\begin{proof} Proof of Lemma \ref{l:pstGmix}.
	Under the observations ${\bf n} = \{n_1,\ldots, n_m\}$,
	\begin{eqnarray*}
	\pi(\lambda | {\bf n}) &=& \frac{\prod_j \frac{e^{-\lambda}\lambda^{n_j}}{n_j!} \left\{p\cdot\frac{\lambda^{\alpha_1-1}\beta^{\alpha_1} e^{-\beta \lambda}}{\Gamma(\alpha_1)} + (1-p)\cdot\frac{\lambda^{(\alpha_1+\alpha_2)-1}\beta^{\alpha_1+\alpha_2} e^{-\beta \lambda}}{\Gamma(\alpha_1+\alpha_2)}\right\}}{{\displaystyle\int_\Lambda} \prod_j \frac{e^{-\lambda}\lambda^{n_j}}{n_j!} \left\{p\cdot\frac{\lambda^{\alpha_1-1}\beta^{\alpha_1} e^{-\beta \lambda}}{\Gamma(\alpha_1)} + (1-p)\cdot\frac{\lambda^{(\alpha_1+\alpha_2)-1}\beta^{\alpha_1+\alpha_2} e^{-\beta \lambda}}{\Gamma(\alpha_1+\alpha_2)}\right\}\,\rm d\lambda}\nonumber\\[2.5mm]
		%
	%
	&=& \frac{p\cdot \frac{\lambda^{\sum_j n_j +\alpha_1-1}\beta^{\alpha_1} e^{-(m+\beta) \lambda}}{\Gamma(\alpha_1)} + (1-p) \cdot \frac{\lambda^{\sum_j n_j+(\alpha_1+\alpha_2)-1}\beta^{\alpha_1+\alpha_2} e^{-(m+\beta) \lambda}}{\Gamma(\alpha_1+\alpha_2)}}{p\cdot\frac{\Gamma(\sum_j n_j+\alpha_1)}{\Gamma(\alpha_1)}\left(\frac{\beta}{\beta+m}\right)^{\alpha_1}\left(\frac{1}{\beta+m}\right)^{\sum_j n_j}+ (1-p)\cdot \frac{\Gamma(\sum_j n_j+\alpha_1+\alpha_2)}{\Gamma(\alpha_1+\alpha_2)}\left(\frac{\beta}{\beta+m}\right)^{\alpha_1+\alpha_2}\left(\frac{1}{\beta+m}\right)^{\sum_j n_j} }\nonumber\\[2.5mm]
		%
	%
	&=&\frac{1}{1+G(p, \alpha_1, \alpha_2, \beta, {\bf n})}\cdot \frac{ (\beta+m)^{\sum_j n_j +\alpha_1}\lambda^{\sum_j n_j +\alpha_1-1}e^{-(m+\beta) \lambda} }{\Gamma(\sum_j n_j+\alpha_1)} \nonumber
	\\[2.5mm]
		%
	%
	&& \hspace{0.5cm}  + \frac{G(p, \alpha_1, \alpha_2, \beta, {\bf n})}{1+G(p, \alpha_1, \alpha_2, \beta, {\bf n})}\cdot \frac{ (\beta+m)^{\sum_j n_j +\alpha_1+\alpha_2}\lambda^{\sum_j n_j +\alpha_1+\alpha_2-1}e^{-(m+\beta) \lambda} }{\Gamma(\sum_j n_j+\alpha_1+\alpha_2)}.
	\end{eqnarray*}
	The result is as described. \hfill
\end{proof}\\

\begin{proof} Proof of Lemma \ref{l:sevpost}.
	Under the observations ${\bf y} = \{y_1,\ldots, y_{m^\ast}\}$, and let $\Pi_\Theta(\cdot)$ be the distribution function of $\Theta$
	\begin{eqnarray}
	\pi(\lambda | {\bf n}) &=& \frac{\left(\prod_{i=1}^{m^\ast} \mu e^{-\mu y_i}\right) 
		\left\{\nu\Delta_\theta(\{\mu\})\right\}
		+
		\left(\prod_{i=1}^{m^\ast} \theta e^{-\theta y_i}\right)
		 \left\{ (1-\nu) \frac{\theta^{\delta-1} \sigma^{\delta} e^{-\sigma \theta}}{\Gamma(\delta)}\right\}}
	 {\int_\Theta \left(\prod_{i=1}^{m^\ast}  \theta e^{-\theta y_i}\right) \rm d\Pi_\Theta(\theta) }
	 	\nonumber\\
	&& \nonumber\\[0.1cm]
	&=& 
	\frac{\nu \left(\mu^{m^\ast} e^{-\mu \sum_i y_i}\right) \Delta_\theta(\{\mu\}) + (1-\nu) \, \frac{\theta^{{m^\ast}+\delta -1}\sigma^{\delta} e^{-(\sigma+\sum_i y_i) \theta}}{\Gamma({m^\ast}+\delta)}}
	{\nu \, \left(\mu^{m^\ast} e^{-\mu \sum_i y_i}\right) + (1-\nu) \, \frac{\Gamma({m^\ast}+\delta) \sigma^{{m^\ast}+\delta}}{\Gamma(\delta) (\sigma+\sum_i y_i)^{{m^\ast}+\delta}} }\nonumber\\
	&& \nonumber\\[0.1cm]
	&=&
	\frac{1}{1+\varphi(\nu, \mu, \delta, \sigma, {\bf y})} \, \Delta_\theta(\{\mu\}) + \frac{\varphi(\nu, \mu, \delta, \sigma, {\bf y})}{1+\varphi(\nu, \mu, \delta, \sigma, {\bf y})}\, \frac{ (\sigma+\sum_i y_i)^{{m^\ast}+\delta} \, \theta^{{m^\ast} +\delta-1} \, e^{-(\sigma+\sum_i y_i) \theta} }{\Gamma({m^\ast}+\delta)} \,. \nonumber \\
	& & 
	\end{eqnarray}
	The result is as described. \hfill
\end{proof}

\end{document}